\documentclass[preprint]{elsarticle}

\usepackage{lineno,hyperref}

\journal{Journal of Computational Physics}

\biboptions{numbers,sort&compress}

\usepackage{amsmath}
\usepackage{amsfonts}
\usepackage{amssymb}
\usepackage{amsthm}
\usepackage{algorithm}
\usepackage{algcompatible}
\usepackage{tikz}
\usepackage{caption}
\usepackage{subcaption}
\usepackage{amscd,mathrsfs} 
\usepackage{xcolor}
\usepackage[margin=1in]{geometry}
\usepackage{comment}

\newcommand{\tp}[1]{{#1}^{\mathsf T}}
\renewcommand{\bar}{\overline}
\newcommand{\eps}{\epsilon}
\newcommand{\pa}{\partial}

\renewcommand{\eps}{\varepsilon}
\renewcommand{\epsilon}{\varepsilon}
\renewcommand{\Sigma}{\varSigma}

\newcommand{\Co}{\mathcal{C}}
\newcommand{\aCo}{\mathcal{\hat C}}
\newcommand{\Ko}{\mathcal{K}}
\newcommand{\aKo}{\mathcal{\hat K}}
\newcommand{\Ho}{\mathcal{H}}
\newcommand{\aHo}{\mathcal{\hat H}}
\newcommand{\Io}{\mathcal{I}}
\newcommand{\Uo}{\mathcal{U}}
\newcommand{\Vo}{\mathcal{V}}
\newcommand{\Hi}{\mathbb{X}}
\newcommand{\bbR}{\mathbb R}

\makeatletter 
 
\@addtoreset{algorithm}{section} 
\makeatother

\newtheorem{thm}[algorithm]{Theorem}

\newtheorem{cor}[algorithm]{Corollary}

\newtheorem{asump}[algorithm]{Assumption}

\newtheorem{rk}[algorithm]{Remark}
\newcommand{\E}{\mathrm E}

\newcommand{\Cov}{\mathrm{Cov}}

\DeclareMathAlphabet\mathbfcal{OMS}{cmsy}{b}{n}
\newcommand{\Imag}{\mathrm{Im}}

\newcommand{\diff}{\text{d}}

\newcommand{\mean}[1]{\overline{#1}}

\newcommand{\half}{\frac{1}{2}}
\newcommand{\pdiff}[2]{\frac{\partial #1}{\partial #2}}

\algnewcommand{\IfThenElse}[3]{
  \State \algorithmicif\ #1\ \algorithmicthen\ #2\ \algorithmicelse\ #3}

\begin{document}

\begin{frontmatter}

\title{Adaptive Dimension Reduction to Accelerate Infinite-Dimensional Geometric Markov Chain Monte Carlo}



\author[UIUC]{Shiwei Lan\corref{cor}}
\ead{shiwei@illinois.edu}


\cortext[cor]{Corresponding author}

\address[UIUC]{Department of Statistics, University of Illinois Urbana-Champaign, Champaign, IL 61820, USA}

\begin{abstract}
Bayesian inverse problems highly rely on efficient and effective inference methods for uncertainty quantification (UQ).
Infinite-dimensional MCMC algorithms, directly defined on function spaces, are robust under refinement (through discretization, spectral approximation) of physical models. 
Recent development of this class of algorithms has started to incorporate the geometry of the posterior informed by data
so that they are capable of exploring complex probability structures, as frequently arise in UQ for PDE constrained inverse problems.
However, the required geometric quantities, including the Gauss-Newton Hessian operator or Fisher information metric, are usually expensive to obtain in high dimensions. 
On the other hand, most geometric information of the unknown parameter space in this setting is concentrated in an \emph{intrinsic} finite-dimensional subspace. 
To mitigate the computational intensity and scale up the applications of infinite-dimensional geometric MCMC ($\infty$-GMC), 
we apply geometry-informed algorithms to the intrinsic subspace to probe its complex structure,
and simpler methods like preconditioned Crank-Nicolson (pCN) to its geometry-flat complementary subspace. 

In this work, we take advantage of dimension reduction techniques 
to accelerate the original $\infty$-GMC algorithms. 
More specifically, partial spectral decomposition (e.g. through randomized linear algebra) of the (prior or Gaussian-approximate posterior) covariance operator is used to identify certain number of principal eigen-directions as a basis for the intrinsic subspace.
The combination of dimension-independent algorithms, geometric information, and dimension reduction yields more efficient implementation, \emph{(adaptive) dimension-reduced infinite-dimensional geometric MCMC}. 
With a small amount of computational overhead, we can achieve over 70 times speed-up compared to pCN using a simulated elliptic inverse problem and an inverse problem involving turbulent combustion with thousands of dimensions after discretization.
A number of error bounds comparing various MCMC proposals are presented to predict the asymptotic behavior of the proposed dimension-reduced algorithms.
\end{abstract}


\end{frontmatter}



\section{Introduction}

Sampling from posterior distributions in the context of PDE-constrained inverse problems is typically a challenging task due to the high-dimensionality of 
the target, the non-Gaussianity of the posterior and the intensive computation of repeated PDE solutions for evaluating the likelihood function at different parameters. 
Traditional Metropolis-Hastings algorithms are characterized by deteriorating mixing times upon mesh-refinement in the finite-dimensional projection of parameter $u$.
This has prompted the recent development of a class of `dimension-independent' MCMC methods \citep{beskos08, beskos11, cotter13, law2014, beskos14, pinski15, rudolf2018, cui16, constantine2016, beskos2017} that overcome this deficiency.
Compared to traditional Metropolis-Hastings defined on the finite-dimensional space, the new algorithms are well-defined on the infinite-dimensional Hilbert space,
thus yield the important computational benefit of mesh-independent mixing times for the practical finite-dimensional algorithms run on the computer.

Among those works, \cite{beskos14, cui16, beskos2017} incorporate geometry of the posterior informed by data to empower the MCMC to become more capable of exploring complicated distributions that deviate significantly from Gaussian. In particular, infinite-dimensional geometric MCMC ($\infty$-GMC) \cite{beskos2017} put a series of `dimension-independent' MCMC algorithms in the context of increasingly adopting geometry (gradient, Hessian).
With the help of such geometric information, \cite{beskos2017} show that with the prior-based splitting strategy, $\infty$-GMC algorithms can achieve up to two orders of magnitude speed up in sampling efficiency compared to vanilla pCN.
However, fully computing the required geometric quantities is prohibitive in the discretized parameter space with thousands of dimensions.
Therefore, it is natural to consider approximations to the gradient vector and Hessian (Fisher) matrix and compute them in a subspace with reduced dimensions.
The key to the dimension reduction in this setting is to identify an intrinsic low-dimensional subspace and apply geometric methods to effectively explore its complex structure;
while simpler methods can be used on its complementary subspace with larger step sizes.
\cite{cui16} seek the intrinsic subspace, known as \emph{likelihood-informed subspace (LIS)} \citep{cui14}, by detecting the eigen-subspace of some globalized Hessian; 
\cite{constantine2016} investigate the \emph{active subspace (AS)} \citep{constantine14,constantine15} by probing the principal eigen-directions of a prior-averaged empirical Fisher matrix.
More recently, \cite{zahm2018} follow the same spirit to exploit the low-dimensional structure, in which the posterior changes the most from the prior.
Their approach is based on approximating the likelihood function with a ridge function that depends non-trivially only on a few linear combinations of the parameters.
Such ridge approximation is obtained by minimizing an upper bound on the Kullback-Leibler distance between the posterior distribution and its approximation.

In this paper, we propose dimension reduction directly based on partial (generalized) spectral decomposition of the prior covariance or the covariance of local Gaussian approximation to the posterior (GAP).
The intrinsic low-dimensional subspace is identified by $r$ leading eigen-functions, which can be efficiently obtained by randomized linear algebraic algorithms \citep{halko11,saibaba16,liberty07}.
Unlike \cite{cui16}, the posterior covariance projected in the subspace is not empirically updated, but rather approximated in a diagonal form which can still capture the most variation of the projected posterior. The resulting GAP covariance has a low-rank structure. Such approximation can be either adopted position-wise or adapted towards a global LIS within the burn-in stage of MCMC. The latter yields a much simpler yet comparable or even more efficient MCMC algorithm compared to \emph{dimension-independent likelihood Informed MCMC (DILI)} \citep{cui16}. The former can demonstrate advantage in sampling complicated posteriors where a globalized pre-conditioner in DILI does not work well universally.
We apply the same dimension reduction to `Hamiltonian Monte Carlo (HMC)' type algorithms so that they can further suppress the diffusive behavior of `MALA' type algorithms and generalize DILI.
We also provide theoretical bounds for comparing dimension-reduced MCMC proposals with their full versions to help understand their asymptotic behavior as well as their difference.

The contributions of this paper are multi-fold.
First, we accelerate the original $\infty$-GMC methods with dimension reduction techniques to scale their applications in PDE constrained inverse problems up to thousands of dimensions.
Second, based directly on partial spectral decomposition, we propose more efficient methods that simplify and generalize DILI \citep{cui16}.
We also establish interesting connections between our adaptive algorithm and DILI.
Third, we derive theoretic bounds comparing several dimension-independent MCMC proposals to describe their asymptotic behavior.
Lastly, we demonstrate the numerical advantage of our proposed algorithms in the high-dimensional setting by over 70 times speed-up compared to pCN method using a simulated elliptic inverse problem and an inverse problem of turbulent jet.

The rest of the paper is organized as follows. Section \ref{sec:review} reviews the background of Bayesian inverse problems, infinite-dimensional geometric MCMC ($\infty$-GMC) \citep{beskos2017} and dimension-independent likelihood Informed MCMC (DILI) \citep{cui16}. Section \ref{sec:DR} describes the details of dimension reduction we adopt, based on the prior and the GAP posterior respectively.
Section \ref{sec:DR-alg} applies these dimension reduction techniques to $\infty$-GMC to achieve acceleration,
establishes the validity of the proposed methods, and also provides error bounds for comparing various algorithms.
In Section \ref{sec:numerics} we show the numerical advantage of our algorithms using a simulated elliptic inverse problem and an inverse problem involving turbulent combustion.
Finally we make some discussion and conclude with a few future directions in Section \ref{sec:conclusion}.

\section{Review of Background}\label{sec:review}

\subsection{Bayesian Inverse Problems}

In Bayesian inverse problems, the objective is to identify an unknown parameter function $u$ which is assumed to be in a separable Hilbert space $(\Hi, \langle\cdot, \cdot \rangle,\lvert\cdot \rvert)$. 
Given finite-dimensional observations $y\in \mathbb{Y}=\mathbb{R}^{m}$, for $m\ge 1$, $u$ is connected to $y$ through the following mapping:
\begin{equation*}\label{eq:forward}
y=\mathcal G(u) + e\ ,\quad e\sim f \ , 
\end{equation*}
where $\mathcal G:\Hi\mapsto \mathbb{Y}$ is the forward operator that maps unknown parameter $u$ onto the data space $\mathbb{Y}$, 
and $f$ is the distribution of noise $e$.
If we assume the density of noise distribution, still denoted as $f$, to exist with respect to the Lebesgue measure,
then we can define the negative log-likelihood, a.k.a. potential function, $\Phi:\Hi\times \mathbb{Y}\to \mathbb R$ as:
\begin{equation*}
\Phi(u;y)=-\log f\big\{\big(y-\mathcal{G}(u)\big);u\big\} \ , 
\end{equation*}
with $f\{\cdot\,;u\}$ indicating the density function for a given $u$.
The noise distribution $f$ could be simple, but the forward mapping $\mathcal G$ is usually non-linear thus the potential function $\Phi(u)$ can be complicated and computationally expensive to evaluate.
For example, if we assume Gaussian noise $e \sim \mathcal{N}_m(0,\Sigma)$, for some symmetric, positive-definite 
$\Sigma\in \mathbb{R}^{m\times m}$, then the potential function $\Phi$ can be written as
\begin{equation}\label{eq:gauss_nz}
\Phi(u;y) = \tfrac{1}{2}\,\big|y-\mathcal{G}(u)\big|^{2}_{\Sigma}
\end{equation}
where we have considered the scaled inner product $\langle \cdot,\cdot \rangle_{\Sigma} = \langle \cdot, \Sigma^{-1}\cdot \rangle$.

In the Bayesian setting, a prior measure $\mu_0$ is imposed on $u$. 
In this paper we assume a Gaussian prior $\mu_0 = \mathcal N(0,\mathcal C)$ with the covariance $\mathcal C$ being a positive, self-adjoint and trace-class operator on $\Hi$.
Now we can get the posterior of $u$, denoted as $\mu^y$, using Bayes' theorem \citep{stuart10,dashti2017}: 
\begin{equation*}
\frac{d\mu^y}{d\mu_0}(u) = \frac{1}{Z}\,\exp(-\Phi(u;y)) \ , \quad \textrm{if} \ 0< Z:=\int_{\Hi} \exp(-\Phi(u;y)) \mu_0(du) < +\infty \ .
\end{equation*}
Notice that the posterior $\mu^y$ can exhibit strongly non-Gaussian behavior,
with finite-dimensional projections having complex non-elliptic contours. 

For simplicity we drop  $y$ from terms involved, so we denote the posterior as $\mu(du)$ and the potential function as $\Phi(u)$.
For the target $\mu(du)$ and many proposal kernels $Q(u,du')$ in the sequel, we define the bivariate 
law: 
\begin{equation*}
\label{eq:MH1}
\nu(du,du') = \mu(du)\,Q(u,du')\ . 
\end{equation*}
Following the theory of Metropolis-Hastings on general spaces \citep{tierney98},
the acceptance probability 
$a(u,u')$ is non-trivial when 
\begin{equation*}
\label{eq:MH2}
\nu(du,du')\simeq\nu^{\top}(du,du'):=\nu(du',du)\  .
\end{equation*}
where $\simeq$ denotes mutual absolute continuity, that is, $\nu \ll \nu^{\top}$ and $\nu^{\top} \ll \nu$.
The acceptance probability is:
\begin{equation}
\label{eq:MH3} 
a(u,u') = 1\wedge \frac{d\nu^{\top}}{d\nu}(u,u')\ . 
\end{equation}
where $\alpha \wedge \beta$ denotes the minimum of $\alpha, \beta\in \mathbb{R}$.
We first review $\infty$-GMC \citep{beskos2017}.

\subsection{Infinite-Dimensional Geometric MCMC ($\infty$-GMC)}
\label{sec:geo}

We start with the \emph{preconditioned Crank-Nicolson (pCN)} method \citep{neal10,beskos08,cotter13}, whose proposal does not use any data information. 
It modifies standard random-walk Metropolis (RWM) to make a proposal movement from the current position towards a random point,
with its size controlled by a free parameter $\rho\in [0,1)$:
\begin{equation*}
u'=\rho\,u+\sqrt{1-\rho^2}\,\xi\ , \quad \xi \sim \mathcal N(0, \mathcal C)
\end{equation*}
PCN is well-defined on the Hilbert space $\Hi$ with the proposal that preserves the prior when $\Phi\equiv 0$, whereas standard RWM 
can only be defined on finite-dimensional discretization and has diminishing 
acceptance probability for fixed step-size and increasing resolution \citep{roberts97}. 
Thus, pCN mixes faster than RWM in high-enough dimensions and the disparity in mixing rates becomes greater upon mesh-refinement \citep{cotter13}.

One approach for developing data-informed methods is to take advantage of gradient information in a steepest-descent setting.
Consider the Langevin SDE on $\Hi$, preconditioned by some operator $\Ko(u)$:
\begin{equation}\label{eq:Langevin}
\frac{du}{dt} = -\frac12\,\Ko(u)\,\big\{ \Co^{-1}u+ \nabla\Phi(u)\big\} + \sqrt{\Ko(u)}\, \frac{dW}{dt}
\end{equation}
with $\nabla\Phi(u)$ denoting the Fr\'echet derivative of $\Phi$ 
and $W$ being the cylindrical Wiener process. 
If we let $\Ko(u)\equiv\Co$, the scales of these dynamics are tuned to the prior. 
In this setting, SDE \eqref{eq:Langevin} preserves the posterior $\mu$ and  
can be used as effective MCMC proposals \citep{beskos08,cotter13}.
\cite{beskos08} use a semi-implicit Euler scheme to discretize the above SDE
and develop \emph{infinite-dimensional MALA ($\infty$-MALA)} with the following proposal
for an algorithmic parameter $\alpha\equiv 1$ and some small step-size $h>0$:
\begin{equation}
\label{eq:infMALA}
\begin{aligned}
u'&=\rho\,u + \sqrt{1-\rho^2}\,\tilde u\ , \quad \tilde u= \xi-\tfrac{\alpha\sqrt{h}}{2}\,\mathcal{C} \nabla\Phi(u) \ , \quad 
\rho = (1-\tfrac{h}{4})/(1+\tfrac{h}{4})\ . 
\end{aligned}
\end{equation}
Following \cite{beskos08}, under the assumption that $\Co \nabla \Phi(u)\in 
\mathrm{Im}(\Co^{\half})$, $\mu_0$-a.s.\@ in~$u$, one can use Theorem 2.21 of \cite{da14} on translations 
of Gaussian measures on separable Hilbert spaces, to obtain 
the Radon-Nikodym derivative in the acceptance probability \eqref{eq:MH3}.

To further incorporate local geometric information of the target distribution, one can consider 
a location-specific pre-conditioner $\Ko(u)$ as the covariance of a local Gaussian approximation $\mathcal{N}(m(u),\Ko(u))$ to the posterior, 
hence named Gaussian-approximate posterior (GAP) covariance,
defined through
\begin{equation}\label{eq:metric}
\Ko(u)^{-1} = \Co^{-1} + \beta \Ho(u) ,
\end{equation}
where $\Ho(u)$ can be chosen as Hessian, Gauss-Newton Hessian (GNH), or Fisher information operator.
With the Gaussian likelihood \eqref{eq:gauss_nz}, we note that they are connected as follows
\begin{align*}
\nabla^2\Phi(u) &= \langle \nabla\mathcal G(u),  \nabla\mathcal G(u)\rangle_\Sigma + \langle \nabla^2\mathcal G(u),  \mathcal G(u)-y \rangle_\Sigma\\
\underbrace{\langle \nabla\mathcal G(u),  \nabla\mathcal G(u)\rangle_\Sigma}_{GNH} &= \E_{y|u}[\underbrace{\nabla^2\Phi(u;y)]}_{Hessian} = \underbrace{\E_{y|u}[\nabla\Phi(u;y) \otimes \nabla\Phi(u;y)]}_{Fisher}
\end{align*}
In the following, unless stated otherwise, we will use GNH for $\Ho(u)$. 
Then $\Ko(u)^{-1}$ can be viewed as the Gauss-Newton Hessian approximation to the log-posterior.
In general $\Ko(u)^{-1}$ defines a Riemannian metric on the parameter space $\Hi$ which thus can be viewed as a Riemannian manifold \citep{girolami11}.
Notice that for $\beta\equiv 1$ the resulting dynamics do not, in general, preserve
the target $\mu$ as they omit the 
higher order (and computationally expensive) Christofell symbol terms, see e.g.\@
\cite{girolami11} and the discussion in 
\cite{xifara:14}.
However, the dynamics in \eqref{eq:Langevin}
can still capture an important part of the local curvature structure of the target and can provide an effective balance between mixing and computational cost \citep{girolami11}.
\cite{beskos2017} develop \emph{infinite-dimensional manifold MALA ($\infty$-mMALA; the name comes from the fact that the Langevin SDE \eqref{eq:Langevin} is defined on the manifold $\langle \Hi, \Ko(u)^{-1}\rangle$.)}
with the following proposal obtained by the similar semi-implicit scheme as 
in \cite{beskos08}
\begin{equation}\label{eq:infmMALA}
u' = \rho\,u + \sqrt{1-\rho^2} \,\tilde u\ , \quad \tilde u=\xi +  \tfrac{\sqrt{h}}{2} g(u)\ ,\quad  \xi\sim \mathcal N(0,\Ko(u))\ ,
\end{equation}
for $\rho$ defined as in (\ref{eq:infMALA}) and we have:
\begin{equation}\label{eq:ngrad}
g(u) = -\Ko(u) \big\{ \alpha \nabla\Phi(u) - \beta \Ho(u) u \big\}\  .
\end{equation}
With the assumptions 3.1-3.3 in \cite{beskos2017}, one can use the Feldman-Hajek theorem (see e.g.\@ Theorem 2.23 in \cite{da14})
to derive the acceptance probability \eqref{eq:MH3}. See more details in \cite{beskos2017}.
It is interesting to notice that
when $\rho=0$ ($h=4$), $\infty$-mMALA coincides with the \emph{stochastic Newton (SN)} MCMC method \citep{martin12},
with $\Ho(u):=\nabla^2\Phi(u)$.

One can generalize `MALA' type algorithms to multiple steps.
This is equivalent to investigating the following continuous-time Hamiltonian dynamics:
\begin{equation}
\label{eq:mHamiltonian0}
\frac{d^2u}{dt^2} + \Ko(u)\,
\big\{\, \mathcal{C}^{-1}u + \nabla\Phi(u) \big\} = 0, \quad \left. \frac{du}{dt}\right|_{t=0} \sim\mathcal{N}(0,\Ko(u))\ .
\end{equation}
HMC algorithm \citep{neal10} makes use of the Strang splitting scheme
to develop the following St\"ormer-Verlet symplectic integrator \citep{verlet67},
for $g$ as defined in \eqref{eq:ngrad}:
\begin{equation}\label{eq:mHDdiscret}
\begin{aligned}
\tilde u^- &= \tilde u_0 + \tfrac{\epsilon}{2}\,g(u_0)\ ; \\
\begin{bmatrix} u_\epsilon\\ \tilde u^{+}\end{bmatrix} &= \begin{bmatrix} \cos\epsilon & \sin\epsilon\\ -\sin\epsilon & \cos\epsilon
\end{bmatrix}  \begin{bmatrix} u_0\\ \tilde u^{-}\end{bmatrix}\  ;\\
\tilde u_\epsilon &= \tilde u^{+} + \tfrac{\epsilon}{2}\,g(u_\epsilon)\  .
\end{aligned}
\end{equation}
%
Equation \eqref{eq:mHDdiscret} gives rise to the
leapfrog map $\Psi_\epsilon: (u_{0},\tilde u_{0})\mapsto (u_{\epsilon}, \tilde u_{\epsilon})$.
Given a time horizon $\tau$ and current position 
$u$, the MCMC mechanism proceeds 
by concatenating $I=\lfloor \tau/\epsilon \rfloor$ steps of leapfrog map consecutively:
\begin{equation*}
u' =\mathcal{P}_u\big\{{\Psi}_{\epsilon}^{I}(u,\tilde u)\big\}\ , \quad \tilde u\sim\mathcal{N}(0,\Ko(u))
\ . 
\end{equation*}
where $\mathcal{P}_u$ denotes the projection onto the $u$-argument.
For $\alpha\equiv 1$, this yields \emph{infinite-dimensional HMC ($\infty$-HMC)} \citep{beskos11} with $\beta\equiv 0$,
and \emph{infinite-dimensional manifold HMC ($\infty$-mHMC)} \citep{beskos2017} with $\beta\equiv 1$ respectively.
The well-posedness of these `HMC' type algorithms can be established under the same assumptions 3.1-3.3 of \cite{beskos2017}.
The following diagram illustrates graphically the connections between the various algorithms.
\begin{alignat*}{7}
&\boxed{\begin{aligned}&\text{\scalebox{1}{pCN}}\\[-.3\baselineskip]&\text{\scalebox{.6}{$\alpha=0$,\; $\beta=0$}}\\[-.2\baselineskip]\end{aligned}} \xrightarrow{\textrm{gradient}} &&
\boxed{\begin{aligned}&\text{\scalebox{1}{$\infty$-MALA}}\\[-.3\baselineskip]&\text{\scalebox{.6}{$\alpha=1$,\; $\beta=0$}}\\[-.2\baselineskip]\end{aligned}} &&\xrightarrow{\textrm{position-dependent preconditioner}\; \Ko(u)} 
\boxed{\begin{aligned}&\text{\scalebox{1}{$\infty$-mMALA}}\\[-.3\baselineskip]&\text{\scalebox{.6}{$\alpha=1$,\; $\beta=1$}}\\[-.2\baselineskip]\end{aligned}} &&\xrightarrow{h=4} 
\boxed{\begin{aligned}&\text{\scalebox{1}{SN}}\\[-.3\baselineskip]&\text{\scalebox{.6}{$\alpha=1$,\; $\beta=1$,\; $\rho=0$}}\\[-.2\baselineskip]\end{aligned}}\\
& &&\phantom{\qquad}\rotatebox[origin=c]{270}{$\xrightarrow{\textrm{multiple steps}\; (I>1)}$} &&\phantom{\xrightarrow{\textrm{position-dependent preconditioner}\; (\Ko(u))}\qquad} \rotatebox[origin=c]{270}{$\xrightarrow{\textrm{multiple steps}\; (I>1)}$} &&\\
& &&\boxed{\begin{aligned}&\text{\scalebox{1}{$\infty$-HMC}}\\[-.3\baselineskip]&\text{\scalebox{.6}{$\alpha=1$,\; $\beta=0$}}\\[-.2\baselineskip]\end{aligned}} &&\xrightarrow{\textrm{position-dependent preconditioner}\; \Ko(u)} 
\boxed{\begin{aligned}&\text{\scalebox{1}{$\infty$-mHMC}}\\[-.3\baselineskip]&\text{\scalebox{.6}{$\alpha=1$,\; $\beta=1$}}\\[-.2\baselineskip]\end{aligned}} &&\\
\end{alignat*}

\subsection{Dimension-Independent Likelihood Informed MCMC (DILI)}
Now we review the MCMC algorithm DILI \citep{cui16}, which is much relevant to our proposed methods.
The idea of DILI is to separate a low-dimensional LIS where likelihood-informed methods are applied to make inhomogeneous proposal to exploit the posterior structure that deviates from the prior structure; while the complementary space can be efficiently explored with simpler prior-based methods.

Inspired by the low-rank approximations to the Hessian of log-posterior in \cite{martin12,spantini2015}, DILI \citep{cui16} obtains the intrinsic low-dimensional LIS by comparing the Hessian of log-likelihood with prior covariance to identify directions in parameter space along which the posterior distribution differs most strongly from the prior.
DILI also uses the Langevin equation \eqref{eq:Langevin} preconditioned by an operator $\Ko_0:=\Cov_{\mu} [u]$ as the proposal kernel.
However, strictly speaking, the preconditioner $\Ko_0$ is not the same as the location-specific $\Ko(u)$, but rather globalized to aggregate the local geometry informed by data.
More specifically, DILI considers the following prior-preconditioned Gauss-Newton Hessian (ppGNH):
\begin{equation*}\label{eq:ppGNH}
\Co^{\frac12} \circ \Ho(u) \circ \Co^{\frac12}, \quad \Ho(u) := \langle \nabla\mathcal G(u),  \nabla\mathcal G(u)\rangle_\Sigma\  .
\end{equation*}
where $\Ho(u)=\nabla \mathcal{G}(u)^* \Sigma^{-1} \nabla \mathcal{G}(u)$ under the assumption \eqref{eq:gauss_nz}, coincides with Fisher metric.

ppGNH stems from the local Rayleigh ratio $R(\varphi; u) := \frac{\langle \varphi, \Ho(u) \varphi \rangle}{\langle \varphi, \Co^{-1} \varphi \rangle}, \quad \varphi \in \Imag(\Co^{\frac12})$
that quantifies how strongly the likelihood constrains variation in the $\varphi$ direction relative to the prior,
and can be converted to GNH w.r.t the whitened parameter $v:=\Co^{-\frac12} u$
\begin{equation*}\label{eq:ppGNH2}
\Ho(v) = \Co^{\frac12} \Ho(u) \Co^{\frac12}
\end{equation*}
Therefore by transforming $u\mapsto v$ and applying $\Co^{-\frac12}$ on both sides it helps to simplify \eqref{eq:Langevin} with $\Ko_0$ into
\begin{equation}\label{eq:Langevin3}
\frac{dv}{dt} = -\frac12\,\Ko\,\big\{ \Io v+ \nabla_v\Phi(v)\big\} +\, \sqrt{\Ko} \frac{dW}{dt}
\end{equation}
where $\nabla_v\Phi(v) = \Co^{\frac12} \nabla_u\Phi(u)$ and $\Ko:=\Co^{-\half} \Ko_0 \Co^{-\half}=\Cov_{\mu} [v]$.
Note the whitened variable $v$ has the prior $\mu_0^v = \mathcal N(0, \Io)$, where the identity covariance operator is not a trace-class on $\Hi$.
However, random draws from $\mu_0^v$ are square-integrable in the weighted space $\Imag(\Co^{-\frac12})$.
\eqref{eq:Langevin3} can still serve as a well-defined function space proposal for parameter $u$ after inverting the transformation.

The intrinsic low-dimensional subspace is obtained through a low-rank approximation of the globalized (expected) GNH $\Ho$. 
Suppose the operator $\Ho$ has eigen-pairs $\{\lambda_i, v_i(x)\}$ on $\Imag(\Co^{-\frac12})$. Then by thresholding $r$ largest eigenvalues one can define
\begin{subequations}\label{eqs:proj2r}
\begin{alignat}{2}
\Vo_r : \bbR^r \rightarrow \Hi, &\quad \{w_i\}_i^r \mapsto \sum_{i=1}^r w_i v_i(x) \\
\Vo_r^* : \Hi \rightarrow \bbR^r, &\quad v(x) \mapsto \{\langle v_i(x), v(x) \rangle\}_{i=1}^r
\end{alignat}
\end{subequations}

Note $\{v_i(x)\}_{i=1}^r$ provides the basis for the LIS and one can have the following decomposition for $v\in\Imag(\Co^{-\frac12})$:
\begin{equation*}\label{eq:decomp}
v = \Vo_r \Vo_r^* v + (\Io - \Vo_r \Vo_r^*) v
\end{equation*}
where $\Vo_r \Vo_r^* v$ is the projection of $v$ into LIS and $(\Io - \Vo_r \Vo_r^*) v$ lies in the complementary space dominated by the prior $\mu_0$;
and they are independent under the approximated posterior $\tilde\mu_y(dv) \propto f(\Vo_r \Vo_r^* v; y) \mu_0(dv)$.
Therefore one can approximate the posterior covariance (for the parameter $v$) as follows
\begin{equation}\label{eq:LISapx}
\Cov_{\mu} [v] \approx \aKo:= \Cov_{\tilde\mu} [v] = \Cov_{\mu} [\Vo_r \Vo_r^* v] + \Cov_{\mu_0} [(\Io - \Vo_r \Vo_r^*) v] 
= \Vo_r K_r \Vo_r^* + \Io - \Vo_r \Vo_r^* = \Psi_r (D_r-I_r) \Psi_r^* + \Io
\end{equation}
where $K_r:=\Cov_{\mu} [\Vo_r^* v]$ is computed empirically and has eigendecomposition $K_r = W_r D_r W_r^*$; $\Psi_r:=\Vo_r W_r$.
We can associate the complement of $\Imag(\Psi_r)$ in $\Imag(\Co^{-\half})$ with a set of eigenfunctions $\{\psi_i(x)\}_{i>r}$. 
Define $\Psi_\perp^*(v(x)):=\{\langle \psi_i(x), v(x) \rangle\}_{i>r}$.
By applying $\Psi_r^*$ and $\Psi_\perp^*$ respectively to \eqref{eq:Langevin3} with $\Ko$ replaced by $\aKo$ in \eqref{eq:LISapx} we obtain the splitting proposal as follows:
\begin{subequations}\label{eqs:DILI-core}
\begin{alignat*}{2}
dw_r &= -\frac12 D_r w_r dt - \frac{\gamma_r}{2} D_r \nabla_{w_r} \Phi(w; y) dt + \sqrt{D_r} dW_r \\
dw_\perp &= -\frac12 w_\perp dt - \frac{\gamma_\perp}{2} \nabla_{w_\perp} \Phi(w; y) dt + dW_\perp 
\end{alignat*}
\end{subequations}
where $w_r = \Psi_r^* v$ and $w_\perp  = \Psi_\perp^* v$; $\gamma_r$ and $\gamma_\perp$ are algorithmic parameters to indicate whether (set to 1) or not (set to 0) to include gradient information in the intrinsic subspace and its complement respectively.

Finally, we apply the semi-implicit Euler scheme as in \cite{beskos08} to the above SDE to get the discrete proposal in $w$ and rewrite it to the following proposal in $v$ (See more details in \cite{cui16} or in Section \ref{sec:DR-infmMALA})
\begin{equation}\label{eq:dili-prop}
\begin{aligned}
v' &= A v - G \nabla_v \Phi(v) +  B \xi \\
A &= \Psi_r (D_{A_r} - a_\perp I_r) \Psi_r^* + a_\perp \Io, \quad B = \Psi_r (D_{B_r} - b_\perp I_r) \Psi_r^* + b_\perp \Io, \quad G = \Psi_r D_{G_r} \Psi_r^*
\end{aligned}
\end{equation}
with the following parameters in the above equation:
\begin{equation}\label{eq:dili-pars}
\begin{aligned}
& D_{A_r} =(I_r-h_r D_r)+( 2I_r+h_r D_r)^{-1}h_r^2D_r^2(1-\gamma_r),\; && 
D_{B_r} =  \sqrt{( 2I_r+h_r D_r)^{-2}8h_rD_r(1-\gamma_r)+2h_rD_r\gamma_r},\\ 
& D_{G_r} = h_r D_r \gamma_r &&
 a_\perp = (2-h_\perp)/(2+h_\perp),\quad b_\perp =  \sqrt{8h_\perp}/(2+h_\perp) &&
\end{aligned}
\end{equation}

In the next section, we will derive similar intrinsic subspaces and splitting proposals directly from partial spectral decomposition.
Based on prior or GAP covariance operators, different dimension reduction strategies can be achieved, applicable to different scenarios depending on how informative the data are.


\section{Dimension Reduction} \label{sec:DR}

We focus on dimension reduction through partial spectral decomposition.
Suppose we have eigen-pairs $\{\lambda_i, \phi_i(x)\}$ of some operator (prior covariance or GAP covariance) with $\lambda_1\geq \lambda_2\geq \cdots$.
The intrinsic low-dimensional subspace can be defined through principal eigen-functions, that is, $\Hi_r=\overline{\{\phi_i(x)\}_{i=1}^r}$. 
Let $\Hi=\Hi_r \oplus \Hi_\perp$.
Then we define the following generic projection operator $\mathcal P_r$, e.g. $\Vo_r \Vo_r^*$ in \eqref{eqs:proj2r}:
\begin{equation*}\label{eq:truncation}
\mathcal P_r: \Hi \to \Hi_r, \quad u\mapsto u^r := \sum_{i=1}^r \phi_i\langle \phi_i,u\rangle \ .
\end{equation*}
For example, if we truncate $\Ho(u)$ on the $r$-dimensional subspace $\Hi_r\subset \Hi$
\begin{equation*}
\Ho_r (u) (w,w') := \langle w, \mathcal P_r^* \Ho (u) \mathcal P_r w' \rangle = \langle \mathcal P_r w, \E_{Y|u}[\nabla_r \Phi(u)\tp{\nabla_r\Phi(u)}] \mathcal P_r w'\rangle, \quad \forall w, w'\in \Hi \ ,
\end{equation*}
where $\nabla_r:=\nabla_{u^r}$ is the restriction of $\nabla$ on $\Hi_r$,
then we can approximate $\Ko(u)$ by replacing $\Ho(u)$ with $\Ho_r(u)$ in \eqref{eq:metric}.
In this section we investigate two types of dimension reduction based on prior and likelihood respectively, with particular connection to DILI \citep{cui16}.

\subsection{Prior-Based Dimension Reduction}
Let $\{\lambda_i\}_{i\ge 1}$, $\{u_i(x)\}_{i\ge 1}$ be the eigenvalues and eigenfunctions of the prior covariance  
operator $\Co$ such that $\Co u_i(x) = \lambda_i u_i(x)$, $i\ge 1$.
Assume $\{\lambda_i\}_{i\ge 1}$ is a sequence of positive reals with $\sum{\lambda_i}<\infty$ (this enforces 
the trace-class condition for $\Co$), and $\{u_i(x)\}_{i\ge 1}$ is an orthonormal basis of $\Hi$.
We make the usual correspondence  
between an element $u^*\in\Hi$ and its coordinates w.r.t.\@ the basis 
$\{u_i(x)\}_{i\ge 1}$, that is
$u^*=\sum_i u^*_i u_i(x) \leftrightarrow \{u^*_i\}_{i\ge 1}$.
Using the Karhunen-Lo\`eve expansion 
of a Gaussian measure \citep{adler10,bogachev98,dashti2017} we have the representation:
\begin{equation}
\label{eq:KL}
u^* \sim \mathcal{N}(0,\Co)\,\,\, \Longleftrightarrow\,\,\,
u^*=\sum_{i=1}^{\infty} u^*_i u_i(x)\ ,\,\, 
u^*_i\sim \mathcal{N}(0,\lambda_i)\ , \,\,\forall \, i\in \mathbb N\ .
\end{equation}
Define the Sobolev spaces 
corresponding to the basis $\{u_i(x)\}$:
\begin{equation*}
\Hi^s = \big\{\{u^*_i\}_{i\ge 1}: \sum i^{2s}|u^*_i|^{2}<\infty  \big\}\ , \quad s\in \mathbb{R}\ ,
\end{equation*}
so that  $\Hi^0\equiv \Hi$ and $\Hi^{s}\subset \Hi^{s'}$ if $s'<s$.
Typically, we will have the following decay assumption:
\begin{asump}\label{asump:eig_decay}
$\sqrt{\lambda_i} = \Theta(i^{-\kappa})$ for some $\kappa>1/2$ 
in the sense that there exist constants $C_1, C_2>0$ such that $C_1\cdot i^{-\kappa}\leq 
\sqrt{\lambda_i}\leq C_2 \cdot i^{-\kappa}$ for all $i \ge 1$.
\end{asump}
Thus, the prior (so also the posterior) concentrate on $\Hi^s$ for any $s<\kappa-1/2$.
Notice also that:
$\mathrm{Im}(\Co^{1/2})= 
\Hi^{\kappa}\ .$
By thresholding $r$ largest eigenvalues $\lambda_1\geq \lambda_2\geq \cdots \lambda_r$, 
eigen-basis $\{u_i(x)\}_{i=1}^r$ can also serve as a basis for low-dimensional subspace.
One can define the projection $\mathcal P_r$ based on the eigen-pairs $\{\lambda_i, \phi_i(x)\}_{i=1}^r$.
Such prior based dimension reduction can work well in a class of inverse problems, especially when they are prior dominated \citep[See][for more details]{beskos2017}.

With the eigen-pairs $\{\lambda_i, u_i(x)\}$, the prior covariance operator $\Co$ can be written and approximated as
\begin{equation*}
\Co = \Uo \Lambda \Uo^* \approx \aCo := \Uo_r \Lambda_r \Uo_r^*
\end{equation*}
with $\Uo_r$ and $\Uo_r^*$ defined by $\{u_i(x)\}_{i=1}^r$ similarly as \eqref{eqs:proj2r}.
Then we can approximate $\Ko(u)$
\begin{equation*}
\Ko(u) = (\Co^{-1} + \Ho(u) )^{-1} \approx \Co^{\half} ( \Io + \aCo^{\half} \Ho(u) \aCo^{\half} )^{-1} \Co^{\half}
= \Co^{\half} ( \Io + \Uo_r \hat H_r(u) \Uo_r^* )^{-1} \Co^{\half}
\end{equation*}
where $\hat H_r(u):= \Lambda_r^{\half} H_r(u) \Lambda_r^{\half}$, and $H_r(u):= \Uo_r^* \Ho(u) \Uo_r$.
By the Sherman-Morrison-Woodbury formula,
\begin{equation*}
\begin{split}
\Ko(u) &\approx 
\Co^{\half} [ \Io - \Uo_r (\hat H_r(u)^{-1} + I_r)^{-1} \Uo_r^* ] \Co^{\half} 
= \Co^{\half} [ \Io + \Uo_r ( (\hat H_r(u) + I_r)^{-1} - I_r) \Uo_r^* ] \Co^{\half} \\
&\approx \Co + \aCo^{\half} \Uo_r ( D_r - I_r) \Uo_r^* \aCo^{\half}
= \aKo(u) := \Co +  \Uo_r \Lambda_r^{\half} (D_r - I_r) \Lambda_r^{\half} \Uo_r^* 
\end{split}
\end{equation*}
where $D_r := (\hat H_r(u)+ I_r)^{-1}$.
This implies $\aKo(u)^\half = \Co^\half +  \Uo_r \Lambda_r^{\half} (\sqrt{D_r} - I_r) \Uo_r^*$.
By applying $\Uo_r^*$ and $\Uo_\perp^*$ to \eqref{eq:Langevin} respectively and using the above approximations $\aKo(u),\ \aKo(u)^\half$ we get
\begin{subequations}
\begin{alignat*}{2}
du_r &= -\half D_r u_r dt - \frac{\gamma_r}{2} D_r \nabla_{u_r} \Phi(u; y) dt + \sqrt{D_r} dW_r \\
du_\perp &= -\half u_\perp dt - \frac{\gamma_\perp}{2} \nabla_{u_\perp} \Phi(u; y) dt + dW_\perp 
\end{alignat*}
\end{subequations}
where $u_r = \Lambda_r^{-\half} \Uo_r^* u$ and $u_\perp  = \Lambda_\perp^{-\half} \Uo_\perp^* u$.

\begin{rk}
The eigen-pairs $\{\lambda_i, u_i(x)\}$ can be pre-computed/pre-specified. For each step we only need to calculate $r$-dimensional matrix $\hat H_r(u)$ at $u$, which can be efficiently obtained by Hessian action through adjoint methods in PDE.
\end{rk}

\subsection{Likelihood-Based Dimension Reduction}
When the data are more informative, the above prior based dimension reduction may not perform well and thus we need reduction techniques that are likelihood based, like DILI \citep{cui16} or AS \citep{constantine2016}.
We could consider the following generalized eigen-problem $(\Ho(u), \Co^{-1})$, to find the eigen-pairs $\{\lambda_i, u_i(x)\}$ such that
\begin{equation}\label{eq:geigu}
\Ho(u) u_i(x) =  \lambda_i \Co^{-1} u_i(x)
\end{equation}
which can be shown equivalent to the eigen-problem of ppGNH $\Ho(v)$ for eigen-pairs $\{\lambda_i, v_i(x)\}$
\begin{equation}\label{eq:eigv}
\Ho(v) v_i(x) = \Co^{\half} \Ho(u) \Co^{\half} v_i(x) =  \lambda_i v_i(x)
\end{equation}
with $v_i(x) = \Co^{-\half}u_i(x)$; it can be written as $\Ho(v) = \Vo(v) \Lambda(v) \Vo(v)^*$ where
\begin{subequations}
\begin{alignat*}{2}
\Vo : \ell^2 \rightarrow \Hi, &\quad \{\lambda_i\} \mapsto \sum_{i=1}^\infty \lambda_i v_i(x) \\
\Vo^* : \Hi \rightarrow \ell^2, &\quad v(x) \mapsto \{\langle v_i(x), v(x) \rangle\}_{i=1}^\infty \\
\Lambda : \ell^2 \rightarrow \ell^2, &\quad \{a_i\} \mapsto \{\lambda_i a_i\} 
\end{alignat*}
\end{subequations}
For the convenience of exposition and comparison with DILI, we work with the whitened coordinates $v(x):=\Co^\half u(x)$ in the following. 
To simplify the notation, we also drop some dependence of $v$ where there is no ambiguity, but readers should be reminded that $\Vo$, $\Vo^*$ and $\Lambda$ are defined and approximated point wise $v\in \Hi$.

In the whitened coordinates $v(x)$, we can rewrite $\Ko(v)^{-1}$, the GNH of the log-posterior, as follows
\begin{equation*}\label{eq:postprec}
\Ko(v)^{-1}:=\Co^{\half} \Ko(u)^{-1} \Co^{\half} = \Io + \Ho(v) = \Io + \Vo \Lambda \Vo^* 
\end{equation*}
which has a direct $r$-dimensional low-rank approximation
\begin{equation}\label{eq:lr-apx}
\Ko(v)^{-1} \approx \Io + \aHo(v) = \Io + \Vo_r \Lambda_r \Vo^*_r 
\end{equation}
Note, in the sense of using the low-rank approximation to the Hessian of log-posterior, this approach is more faithful to \cite{martin12,spantini2015} than DILI.
Thus $\Ko(v)$ can be approximated using the Sherman-Morrison-Woodbury formula
\begin{equation}\label{eq:apx_postcov}
\Ko(v) \approx \aKo(v) := ( \Io + \aHo(v) )^{-1}
 = ( \Io + \Vo_r \Lambda_r \Vo^*_r )^{-1} = \Io - \Vo_r (\Lambda_r^{-1} + I_r)^{-1} \Vo_r^*
 = \Io + \Vo_r (D_r -I_r) \Vo_r^*
\end{equation}
where $D_r := (I_r + \Lambda_r)^{-1}$.

\begin{rk}
With the approximation \eqref{eq:apx_postcov} applied to the following $\Ko$, we directly have
\begin{equation}\label{eq:lo-postcov}
K_r =\Cov_{\mu} [\Vo_r^* v] = \Vo_r^* \Cov_{\mu} [v] \Vo_r = \Vo_r^* \Ko \Vo_r 
\approx I_r - (\Lambda_r^{-1} + I_r)^{-1} 
= D_r
\end{equation}
Therefore we can rewrite the approximation \eqref{eq:apx_postcov} analogous to the approximation  \eqref{eq:LISapx} in DILI \citep{cui16}
\begin{equation*}
\Ko(v) \approx \aKo(v) = \Io + \Vo_r (D_r -I_r) \Vo_r^* = \Vo_r K_r \Vo_r^* + \Io - \Vo_r \Vo_r^*
\end{equation*}
In this sense, the approximation \eqref{eq:apx_postcov} is consistent with the low-rank approximation by DILI.
However, we have avoided the empirical computation of $K_r =\Cov_{\mu} [\Vo_r^* v]$.
And since $K_r$ is already in the diagonal form \eqref{eq:lo-postcov}, we have also avoided the rotation $\Psi_r=\Vo_r W_r$ in DILI \citep{cui16}.
\end{rk}

By applying $\Vo_r^*$ and $\Vo_\perp^*$ to \eqref{eq:Langevin3} respectively and using the approximation \eqref{eq:apx_postcov}, we have the following proposal split on the low-dimensional subspace $\Hi_r$ and its complement $\Hi_\perp$:
\begin{subequations}\label{eqs:split-prop}
\begin{alignat}{2}
dv_r &= -\half D_r v_r dt - \frac{\gamma_r}{2} D_r \nabla_{v_r} \Phi(v; y) dt + \sqrt{D_r} dW_r \\
dv_\perp &= -\half v_\perp dt - \frac{\gamma_\perp}{2} \nabla_{v_\perp} \Phi(v; y) dt + dW_\perp 
\end{alignat}
\end{subequations}
where $v_r = \Vo_r^* v$ and $v_\perp  = \Vo_\perp^* v$.


\section{Dimension Reduced Algorithms} \label{sec:DR-alg}

In this section we apply the dimension reduction techniques discussed in Section \ref{sec:DR} to two $\infty$-GMC algorithms, $\infty$-mMALA and $\infty$-mHMC.
Since the prior-based dimension reduction has been implemented in \cite{beskos2017}, we focus on the likelihood-based dimension reduction.
We derive new efficient algorithms and compare them to DILI.

\subsection{Dimension-Reduced $\infty$-mMALA}\label{sec:DR-infmMALA}

We can apply the similar semi-implicit Euler scheme as in \cite{beskos08} 
to \eqref{eqs:split-prop}, 
or equivalently use the approximation \eqref{eq:apx_postcov} in the following whitened proposal, which is a reformulation of manifold MALA proposal \eqref{eq:infmMALA} by the transformation $v(x)=\Co^{-\half} u(x)$
\begin{equation}\label{eq:infmMALA-whiten}
v' = \rho\,v + \sqrt{1-\rho^2} \,\tilde v\ , \quad \tilde v=\tilde \xi +  \tfrac{\sqrt{h}}{2} g(v)\ ,\quad  \tilde \xi \sim \mathcal N(0,\Ko(v))\ ,
\end{equation}
for $\rho$ defined as in \eqref{eq:infMALA} and setting $\beta=1$, and replacing $\alpha$ with $s(\gamma)$ in \eqref{eq:ngrad} we have:
\begin{equation}\label{eq:ngrad-whiten}
g(v) = \Co^{-\half} g(u) = -\Ko(v) \big\{ - \Ho(v)v + s(\gamma) \nabla_v\Phi(v) \big\}\  
= (\Io - \Ko(v)) v - \Ko(v) s(\gamma) \nabla_v\Phi(v)
\end{equation}
where $s(\gamma)=\Vo_r\gamma_r\Vo_r^* + \Vo_\perp\gamma_\perp\Vo_\perp^*$, $\gamma=\begin{bmatrix} \gamma_r \\ \gamma_\perp \end{bmatrix}$ with $\gamma_r$ and $\gamma_\perp$ chosen to be $0$ or $1$ to indicate whether or not to include gradient information on the low-dimensional subspace $\Hi_r$ and its complement $\Hi_\perp$ respectively. 
This proposal \eqref{eq:infmMALA-whiten} can be reformulated as follows
\begin{equation}\label{eq:reform-infmMALA}
v' 
= \rho_0 v + \rho_1 g(v) + \rho_2 \tilde\xi = (\Io -\rho_1 \Ko(v)) v - \rho_1 \Ko(v) \nabla_v\Phi(v) + \rho_2 \sqrt{\Ko(v)} \xi, \quad \xi \sim \mathcal N(0, \Io)
\end{equation}
where $\rho_0=\rho$, $\rho_1=1-\rho$, and $\rho_2=\sqrt{1-\rho^2}$.

With the approximation \eqref{eq:apx_postcov}, it is straightforward to verify
\begin{equation}\label{eq:apx_postcov2}
\sqrt{\Ko(v)} \approx \aKo(v)^{\half} = \Io + \Vo_r ( \sqrt{D_r} -I_r ) \Vo_r^*
\end{equation}
Substituting the approximations \eqref{eq:apx_postcov} \eqref{eq:apx_postcov2} into \eqref{eq:reform-infmMALA}
yields the following proposal
\begin{equation}\label{eq:split-prop2}
v' = (\rho_0\Io - \rho_1 \Vo_r (D_r - I_r) \Vo_r^*) v - \rho_1 (\Io + \Vo_r (D_r - I_r) \Vo_r^*) s(\gamma) \nabla_v\Phi(v) 
+ \rho_2 (\Io + \Vo_r ( \sqrt{D_r} -I_r ) \Vo_r^*) \xi
\end{equation}
With $\Io=\Vo_r\Vo_r^*+\Vo_\perp\Vo_\perp^*$ and $\Vo=[\Vo_r,\; \Vo_\perp]$,
the proposal \eqref{eq:split-prop2} can be rewritten as
\begin{equation*}
v'  = \Vo \begin{bmatrix} I_r - \rho_1 D_r & 0 \\
0 & \rho_0 \Io_\perp \end{bmatrix} \Vo^* v
- \Vo \begin{bmatrix} \rho_1 D_r & 0 \\
0 & \rho_1 \Io_\perp \end{bmatrix} 
\gamma 
\Vo^* \nabla_v\Phi(v)
+ \Vo \begin{bmatrix} \rho_2 \sqrt{D_r}  & 0 \\
0 & \rho_2 \Io_\perp \end{bmatrix} \Vo^* \xi
\end{equation*}

\begin{rk}
If we choose $\gamma_\perp=0$ as DILI \citep{cui16}, then the proposal \eqref{eq:split-prop2} can further be simplified as
\begin{equation}\label{eq:split-prop3}
v' = (\rho_0\Io - \rho_1 \Vo_r (D_r - I_r) \Vo_r^*) v - \rho_1 \Vo_r D_r \gamma_r \Vo_r^* \nabla_v\Phi(v) 
+ \rho_2 (\Io + \Vo_r ( \sqrt{D_r} -I_r ) \Vo_r^*) \xi
\end{equation}
Comparing with the following operator-weighted proposal as in DILI \eqref{eq:dili-prop} (also Equations (17) (36) of \cite{cui16}; note that $\Psi_r=\Vo_r I_r$ when $K_r$ is already in the diagonal form.)
\begin{equation*}
\begin{aligned}
v' &= A v - G \nabla_v \Phi(v) +  B \xi \\
A &= \Vo_r (D_{A_r} - a_\perp I_r) \Vo_r^* + a_\perp \Io, \quad B = \Vo_r (D_{B_r} - b_\perp I_r) \Vo_r^* + b_\perp \Io, \quad G = \Vo_r D_{G_r} \Vo_r^*
\end{aligned}
\end{equation*}
we found that the proposal \eqref{eq:split-prop3} corresponds to the proposal 3.2 (LI-Langevin) of \cite{cui16} (Equation \eqref{eq:dili-pars} with $\gamma_r=1$) with the following parameters in the above equation:
\begin{equation}\label{eq:connection2dili}
\begin{aligned}
& D_{A_r} = I_r-\rho_1 D_r,\; && D_{B_r} = \rho_2 \sqrt{D_r},\; && D_{G_r} = \rho_1 D_r \gamma_r , \\
& a_\perp = \rho_0,\; && b_\perp =  \rho_2 &&
\end{aligned}
\end{equation}
Similar connection to the proposal 3.1 (LI-Prior) of \cite{cui16} (Equation \eqref{eq:dili-pars} with $\gamma_r=0$) can be drawn with $\alpha=\beta=0$, that is, $g(v)\equiv 0$.

The proposal \eqref{eq:split-prop3} is analogous to but different from DILI proposal in two aspects:
(i)\ while $A$, $B$, $G$ in DILI are fixed once the global LIS is obtained, $\Vo_r(v)$, $D_r(v)$ and $\Vo(v)^*$ implicitly depend on location $v\in\Hi$ thus our proposal \eqref{eq:split-prop3} is more general with the freedom of choosing between a position-specific and an adaptively globalized implementation, as detailed in the following;
(ii) inherited from the semi-implicit scheme of $\infty$-mMALA, $\rho_0(h),\, \rho_1(h),\, \rho_2(h)$ are all functions of the step size $h$, therefore there is only one tuning parameter $h$ for the discretized step size in our proposal \eqref{eq:split-prop3}, while DILI has separate step sizes for $\Hi_r$ and $\Hi_\perp$ respectively, which may increase the difficulty of tuning.
\end{rk}
\begin{rk}
For all proposals in DILI, $\gamma_\perp=0$, which corresponds to a `pCN' (gradient-free) type update in $\Hi_\perp$.
However, for computational practice, once the gradient $\nabla_v\Phi(v)$ has been obtained using adjoint method in PDE, we can still take advantage of it for a `manifold-MALA' (with projected gradient) type update in $\Hi_\perp$.
That is, to keep the gradient term ($\gamma_\perp=1$) in the proposal \eqref{eq:split-prop2}, where we did not actually compute $\Vo_\perp\Vo_\perp^*$.
\end{rk}

With the proposal \eqref{eq:split-prop2}, we denote the position-specific implementation as \emph{Dimension-Reduced $\infty$-dimensional manifold MALA (DR-$\infty$-mMALA)}.
To prove the well-posedness of the resulting MCMC algorithm and derive its acceptance probability, we need the following assumption on the forward mapping $\mathcal G$:
\begin{asump}\label{asump:fwd_diff}
For some $\ell \in [0,\kappa -1/2)$, the mappings $\{\mathcal G_k : \Hi^\ell \mapsto \mathbb R, \ 1 \leq k \leq m\}$ are Fr\'echet differentiable on $\Hi^\ell$ with derivatives $\nabla\mathcal G_k\in\Hi^{-\ell}$.
\end{asump}

The following theorem establishes the validity of DR-$\infty$-mMALA.
\begin{thm}\label{thm:DR-infmMALA}
Under the assumptions \ref{asump:eig_decay} and \ref{asump:fwd_diff}, DR-$\infty$-mMALA with proposal \eqref{eq:split-prop2} is well-defined on Hilbert space $\Hi$, with the acceptance probability specified as follows
\begin{equation}\label{eq:acpt_DRinfmMALA}
\begin{aligned}
a(v,v')
=& 1 \wedge \frac{\kappa(v',v)}{\kappa(v,v')}, \qquad 
\kappa(v,v') = \exp\left\{-\Phi(v)\right\}  \times  \lambda(w^*;v) \\
\lambda(w^*;v) 
=& 
\exp\big\{-
\tfrac{h}{8}|\aKo^{-\frac12}(v)\hat g(v)|^2+
\tfrac{\sqrt{h}}{2}\langle \aKo(v)^{-\half}g(v),\aKo(v)^{-\half}w^*\rangle\big\} \\
& \qquad \qquad \times 
 \exp\big\{-\tfrac12\langle w^*, \aHo(v) w^* \rangle \big\}\cdot 
|\,\aKo(v)^{-1/2}\,|\ .
\end{aligned}
\end{equation}
where $w^*=\Co^{-\half} w=\frac{v'-\rho_0 v}{\rho_2}$ and $\hat g(v):=-\aKo(v) \big\{ - \aHo(v)v + s(\gamma) \nabla_v\Phi(v) \big\}$.
\end{thm}
\proof{See \ref{apx:thm-DR-infmMALA}.}
\begin{cor}\label{cor:acpt_compare}
With the setting \eqref{eq:connection2dili}, 
the acceptance probability \eqref{eq:acpt_DRinfmMALA} differs from that of DILI \citep[Equation (40) of][]{cui16} by a determinant term $\frac{|D_r(v)|^\half}{|D_r(v')|^\half}$, which is used to adjust the change of local geometry.
\end{cor}
\proof{See \ref{apx:cor-acpt_compare}.}
\begin{rk}
For the adaptive implementation (detailed below), the extra determinant term could help our proposed method to adapt to the local geometry.
However, once the global LIS is identified, $D_r$ will be fixed at principal eigenvalues, and then these two acceptance probabilities become identical.
\end{rk}

Based on the low-rank approximation to the GNH of the log-posterior, we have obtained a likelihood informed splitting proposal \eqref{eq:split-prop2}.
The low-rank approximation is achieved through partial eigen-decomposition \eqref{eq:geigu} or \eqref{eq:eigv},
which can be efficiently calculated through e.g. Krylov-subspace methods \citep{simpson08} or randomized algorithms \citep{halko11,saibaba16,liberty07}.
However, when the parameter space does not have much variation in curvature, repeated execution of partial spectral decomposition at each location $v$ may outweigh its geometric benefit.
We can adopt the similar adaptation procedure of globalizing GNH \citep{brand02} as in DILI \citep{cui16}.
To distinguish from the location-dependent implementation (DR-$\infty$-mMALA), we name the new algorithm as \emph{adaptive Dimension-Reduced $\infty$-dimensional manifold MALA (aDR-$\infty$-mMALA)} and summarize it in Algorithm \ref{alg:aDR-infmMALA}.
\begin{algorithm}[t]
\caption{Adaptive Dimension-Reduced $\infty$-dimensional manifold MALA (aDR-$\infty$-mMALA)}
\label{alg:aDR-infmMALA}
\centering
\begin{algorithmic}[1]
\REQUIRE During the LIS construction, we retain $(1)\; \{\Lambda_m, \Vo_m\}$ to store the expected GNH evaluated from $m$ samples; and $(2)$ the value of F\"orstner distance $d_{\mathcal F}$ between the most recent two updates of the expected GNH, for LIS convergence monitoring.
\REQUIRE At step $n$, given the state $v_n$, LIS basis $\Vo_r$, and operators $A, B, G$ induced by $\{\Vo_r, D_r, h\}$, one step the algorithm is:
\STATE Compute a candidate $v'=q_n(v_n, \cdot; A, B, G)$ using either LI-prior $(\gamma_r=0)$ or LI-Langevin $(\gamma_r=1)$
\STATE Compute the acceptance probability $a(v_n, v')$
\IfThenElse {$\mathrm{Unif}(0,1]<a(v_n, v')$} {$v_{n+1}=v'$} {$v_{n+1}=v_n$}
\IF{$\mathrm{rem}(n+1, n_{\textrm{lag}})=0$ \& $m<m_{\max}$ \& $d_{\mathcal F}\geq \Delta_{\textrm{LIS}}$}
\STATE $\textrm{UpdateLIS}(\Vo_m, \Lambda_m, v_{m+1}; \Vo_{m+1}, \Lambda_{m+1}, \Vo_{r'}, \Lambda_{r'})$
\STATE Update the LIS convergence diagnostic $d_{\mathcal F}$
\STATE $\Vo_r\leftarrow \Vo_{r'}$, $\Lambda_r\leftarrow \Lambda_{r'}$, $m\leftarrow m+1$
\STATE Update $D_r=(I_r+\Lambda_r)^{-1}$
\STATE Update the operators $\{A, B, G\}$ 
\ENDIF
\end{algorithmic}
\end{algorithm}

\subsection{Dimension-Reduced $\infty$-mHMC}
Similarly as $\infty$-mHMC, one can generalize the above (adaptive) dimension-reduced manifold MALA algorithms to multi-step `HMC' algorithms.
In whitened coordinates, the discretized dynamics \eqref{eq:mHDdiscret} becomes as follows, for $g$ defined as in \eqref{eq:ngrad-whiten}:
\begin{equation}\label{eq:mHDdiscret-whiten}
\begin{aligned}
\tilde v^- &= \tilde v_0 + \tfrac{\epsilon}{2}\,g(v_0)\ ; \\
\begin{bmatrix} v_\epsilon\\ \tilde v^{+}\end{bmatrix} &= \begin{bmatrix} \cos\epsilon & \sin\epsilon\\ -\sin\epsilon & \cos\epsilon
\end{bmatrix}  \begin{bmatrix} v_0\\ \tilde v^{-}\end{bmatrix}\  ;\\
\tilde v_\epsilon &= \tilde v^{+} + \tfrac{\epsilon}{2}\,g(v_\epsilon)\  .
\end{aligned}
\end{equation}
With the same notations, we denote the $\infty$-mHMC proposal in the whitened coordinates as
\begin{equation*}
v' =\mathcal{P}_v\big\{{\Psi}_{\epsilon}^{I}(v,\tilde v)\big\}\ , \quad \tilde v\sim\mathcal{N}(0,\Ko(v))
\ . 
\end{equation*}
Then the acceptance probability (c.f. Algorithm 3.11 $\infty$-mHMC in \cite{beskos2017}) can be reformulated as follows
\begin{equation*}\label{eq:deltaE-whiten}
\begin{aligned}
a(v, v') =& 1 \wedge \exp \left\{- \Delta E(v, v')\right\} \\
\Delta E(v_0,\tilde v_0) 
=&\Phi(v_I)-\Phi(v_0)+\tfrac{1}{2}
\langle \tilde v_I, (\Ko^{-1}(v_I)-\Io)\tilde v_I\rangle -
\tfrac{1}{2}\langle \tilde v_0, (\Ko^{-1}(v_0)-\Io)\tilde v_0\rangle \\[0.2cm]
&-\log|\Ko^{-\half}(v_I)|
+\log|\Ko^{-\half}(v_0)|
-\tfrac{\epsilon^2}{8}\big(\,|g(v_I)|^2-
|g(v_0)|^2\,\big) 
+\tfrac{\epsilon}{2} \sum_{i=0}^{I-1}\big(\,\langle g(v_{i}), \tilde v_i\rangle 
+ \langle g(v_{i+1}),\tilde v_{i+1}%
\rangle\,\big)  \ . 
\end{aligned}
\end{equation*}

With the approximation \eqref{eq:apx_postcov}, one can approximate the gradient $g(v)$ in \eqref{eq:ngrad-whiten} as follows ($\gamma_\perp=0$):
\begin{equation}\label{eq:ngrad-aprx}
g(v) \approx \hat g(v) := (\Io - \aKo(v)) v - \aKo(v) s(\gamma) \nabla_v\Phi(v) = -\Vo_r (D_r-I_r) \Vo_r^* v - \Vo_r D_r \gamma_r \Vo_r^* \nabla_v\Phi(v)
= \Vo_r D_r g_r(v) 
\end{equation}
where $g_r(v):= \Lambda_r \Vo_r^* v - \gamma_r \Vo_r^* \nabla_v \Phi(v) \in \bbR^r$.
Using the above approximation \eqref{eq:ngrad-aprx}, we can rewrite the proposal \eqref{eq:mHDdiscret-whiten} as
\begin{equation}\label{eq:mHDdiscret-aprx}
\begin{aligned}
\tilde v^- &= \tilde v_0 + \tfrac{\epsilon}{2}\,\Vo_r D_r g_r(v_0)\ ; \\
\begin{bmatrix} v_\epsilon\\ \tilde v^{+}\end{bmatrix} &= \begin{bmatrix} \cos\epsilon & \sin\epsilon\\ -\sin\epsilon & \cos\epsilon
\end{bmatrix}  \begin{bmatrix} v_0\\ \tilde v^{-}\end{bmatrix}\  ;\\
\tilde v_\epsilon &= \tilde v^{+} + \tfrac{\epsilon}{2}\,\Vo_r D_r g_r(v_\epsilon)\  .
\end{aligned}
\end{equation}
Based on this proposal, we can develop \emph{Dimension-Reduced $\infty$-dimensional manifold HMC (DR-$\infty$-mHMC)} as a multi-step generalization of DR-$\infty$-mMALA.
The following theorem establishes the validity of DR-$\infty$-mHMC and provides the acceptance probability.
\begin{thm}\label{thm:DR-infmHMC}
Under the assumptions \ref{asump:eig_decay} and \ref{asump:fwd_diff}, DR-$\infty$-mHMC with proposal \eqref{eq:mHDdiscret-aprx} is well-defined on Hilbert space $\Hi$, with the acceptance probability specified as follows
\begin{equation*}\label{eq:deltaE-whiten-aprx}
\begin{aligned}
\Delta E(v_0,\tilde v_0) 
=&\Phi(v_I)-\Phi(v_0)+\tfrac{1}{2}
\Vert \Lambda_r^{\half}(v_I) \Vo_r^* \tilde v_I \Vert^2 -
\tfrac{1}{2}\Vert \Lambda_r^{\half}(v_0) \Vo_r^* \tilde v_0 \Vert^2 \\[0.2cm]
&+ \half \log |D_r(v_I)|
-\half \log |D_r(v_0)|
-\tfrac{\epsilon^2}{8}\big(\,|D_r(v_I)g_r(v_I)|^2-
|D_r(v_0)g_r(v_0)|^2\,\big) \\
&+\tfrac{\epsilon}{2} \sum_{i=0}^{I-1}\big(\,\langle D_r(v_i)g_r(v_{i}), \Vo_r^* \tilde v_i\rangle 
+ \langle D_r(v_{i+1})g_r(v_{i+1}),\Vo_r^* \tilde v_{i+1}%
\rangle\,\big)  \ . 
\end{aligned}
\end{equation*}
\end{thm}
\begin{proof}
It can be proved by closely following Theorem 3.10 of \cite{beskos2017}, with $\Ko(u)$ replaced by $\aKo(v)$ and $g(u)$ replaced by $\hat g(v)$.
\end{proof}
Parallel to aDR-$\infty$-mMALA,
the corresponding algorithm with the adaptation used in DILI, named \emph{adaptive Dimension-Reduced $\infty$-dimensional manifold HMC (aDR-$\infty$-mHMC)}, is the similar to Algorithm \ref{alg:aDR-infmMALA} with $q_n$ being replaced by the above multi-step proposal \eqref{eq:mHDdiscret-aprx}.

Now we give some error bounds for comparing different proposals mentioned above.
\begin{thm}\label{thm:bounds}
Assume that Gauss-Newton Hessian $\Ho(v)$ is a trace-class operator defined on Hilbert space $\Hi$. Then we have the following bounds for the difference between proposals.
\begin{itemize}
\item $\Vert v'_\text{\tiny DR-$\infty$-mMALA} - v'_\text{\tiny $\infty$-mMALA} \Vert \leq 
\begin{cases}
\rho_1 \frac{\lambda_{r+1}}{\lambda_{r+1}+1}( \Vert v\Vert + \Vert \nabla_v\Phi(v)\Vert ) + \rho_2 \frac{\lambda_{r+1}}{\lambda_{r+1}+1 + \sqrt{\lambda_{r+1}+1}} \Vert \xi\Vert , & \gamma_\perp=1\\
\rho_1 \left( \frac{\lambda_{r+1}}{\lambda_{r+1}+1} \Vert v\Vert + \Vert \nabla_v\Phi(v)\Vert \right) + \rho_2 \frac{\lambda_{r+1}}{\lambda_{r+1}+1 + \sqrt{\lambda_{r+1}+1}} \Vert \xi\Vert, & \gamma_\perp=0
\end{cases}$
\item $\Vert v'_\text{\tiny DR-$\infty$-mMALA} - v'_\text{\tiny DILI} \Vert \leq 
\rho_1 \Vert D_r-K_r \Vert_2 (\Vert v\Vert + \Vert \nabla_v\Phi(v)\Vert ) + \rho_2 \Vert D_r^\half-K_r^\half \Vert_2 \Vert \xi\Vert$
\item $\Vert v^I_\text{\tiny DR-$\infty$-mHMC} - v^I_\text{\tiny $\infty$-mHMC} \Vert \leq \mathcal O(\lambda_{r+1})$ if we assume $g(v)$ is Lipschitz continuous and $\gamma_\perp=1$.
\end{itemize}
\end{thm}
\proof{See \ref{apx:thm-bounds}.}
\begin{rk}
This theorem describes the asymptotic behavior of the proposed dimension reduced algorithms when the dimension of intrinsic subspace goes to infinity. It quantifies the differences between the dimension reduced proposals and their full versions, 
which bound the ``loss" in the quality of geometry-informed MCMC proposals by doing dimension reduction.
If $\gamma_\perp=1$, dimension reduced $\infty$-GMC algorithms are asymptotically close to their full versions, i.e.
$\Vert v'_\text{\tiny DR-$\infty$-mMALA} - v'_\text{\tiny $\infty$-mMALA} \Vert \to 0$ and $\Vert v^I_\text{\tiny DR-$\infty$-mHMC} - v^I_\text{\tiny $\infty$-mHMC} \Vert \to 0$ as $r\to \infty$.
Proper $r$ can be chosen by thresholding the last eigenvalue $\lambda_{r+1}$ to some values for local or global LIS' respectively, (e.g. $\rho_\text{l}=\rho_\text{g}=0.01$).
By doing so, one can control the precision of these approximations.
Note that the disparity between DILI and DR-$\infty$-mMALA is not guaranteed to be small, but could be significantly large in their performance. See the example in Section \ref{sec:rans}.
\end{rk}
\begin{rk}
\cite{vogel02} shows that the Hessian of potential function \eqref{eq:gauss_nz} is a compact operator whose range space is independent of mesh resolution,
thus it naturally admits $r$-dimensional low-rank approximation $\Ho(v)$ with $r\leq m$, the size of observations \citep{martin12}. 
In this case, there are at most $m$ non-zero eigenvalues $\lambda_i>0$.
\end{rk}

In this paper, we make use of randomized algorithms for (generalized) eigen-decomposition \citep{halko11,saibaba16,liberty07} for the low-rank approximation \eqref{eq:lr-apx}.
If the dimension of the discretized parameter space is $N$, then dimension reduced algorithms lower the computational cost for these two $\infty$-GMC algorithms from prohibitive cubic scale $\mathcal O(N^3)$ to affordable linear scale $\mathcal O(N(r+p)^2)$, where $p$ is the oversampling size with common choice $p=5$ in practice \citep{halko11}.
Note, in \cite{beskos2017} we sample Karhunen-Lo\`eve coefficients \eqref{eq:KL} of dimension about $100$.
Here we sample the whole field (discretized function defined on a domain) whose dimension could scale up to thousands.
Full implementation of the standard $\infty$-GMC algorithms (e.g. $\infty$-mMALA or  $\infty$-mHMC) will be impractical due to the memory bound and time consumption.
In the next section, we will demonstrate the computational advantage of our proposed dimension reduced algorithms using numerical examples.


\begin{figure}[t]
\begin{subfigure}[b]{1\textwidth}
\includegraphics[width=1\textwidth,height=.28\textwidth]{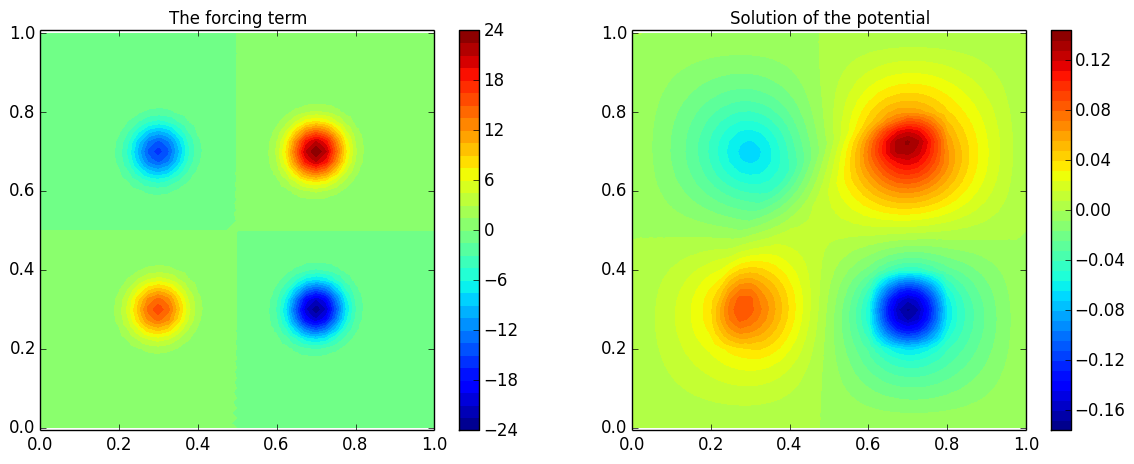}
\caption{Forcing field $f(s)$ (left), and the solution $p(s)$ with true transmissivity field $k_0(s)$ (right).}
\label{fig:force_soln}
\end{subfigure}
\begin{subfigure}[b]{1\textwidth}
\includegraphics[width=1\textwidth,height=.28\textwidth]{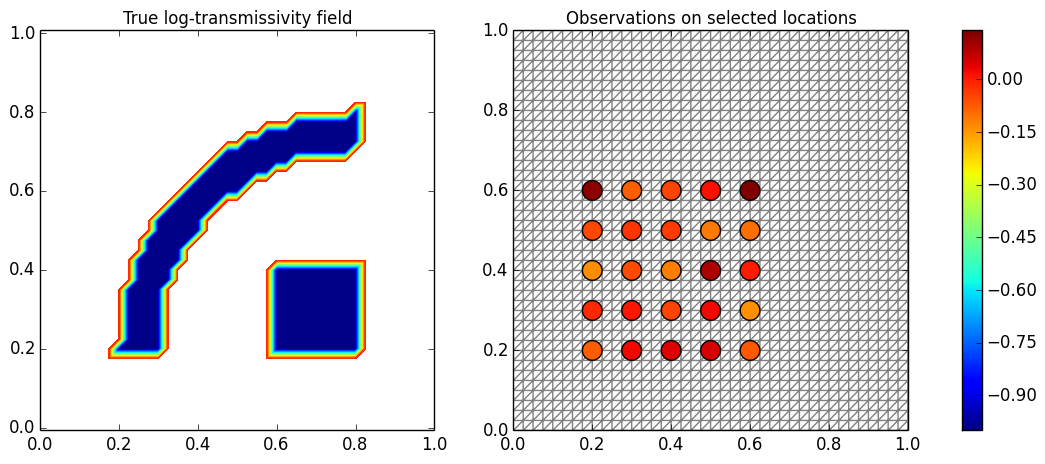}
\caption{True log-transmissivity field $u_0(s)$ (left), and $25$ observations on selected locations indicated by circles (right), with color indicating their values.}
\label{fig:truth_obs}
\end{subfigure}
\end{figure}

\section{Numerical Experiments}\label{sec:numerics}
In this section, we first consider the simulated elliptic inverse problem same as in \cite{cui16} for two cases. This is done for a parallel comparison with DILI.
Then we investigate an inverse problem in Reynolds-Averaged Navier-Stokes (RANS) equation for turbulent combustion.
This problem involves highly nonlinear PDE which is much expensive to solve, and serves as a good benchmark for testing the performance of all the above mentioned algorithms.
Python codes are publicly available at \url{https://bitbucket.org/lanzithinking/dimension-reduced-geom-infmcmc}.

\subsection{Elliptic Inverse Problem}


The following elliptic PDE is defined on the unit square domain $\Omega=[0,1]^2$:
\begin{equation*}\label{eq:elliptic}
\begin{aligned}
-\nabla \cdot (k(s) \nabla p(s)) &= f(s), \; s\in \Omega \\
\langle k(s) \nabla p(s), \vec n(s) \rangle &= 0, \; s \in \pa\Omega \\
\int_{\pa\Omega} p(s) dl(s) &= 0
\end{aligned}
\end{equation*}
where $k(s)$ is the transmissivity field, $p(s)$ is the potential function, $f(s)$ is the forcing term,
and $\vec n(s)$ is the outward normal to the boundary.
The source/sink term $f(s)$ is defined by the superposition of four weighted Gaussian plumes with standard deviation $0.05$, 
centered at $[0.3, 0.3],\, [0.7, 0.3],\, [0.7, 0.7],\, [0.3, 0.7]$, with weights $\{2, -3, 3, -2\}$ respectively, as shown in the left panel of Figure \ref{fig:force_soln}.

\begin{figure}[t]
\begin{center}
\includegraphics[width=1\textwidth,height=.4\textwidth]{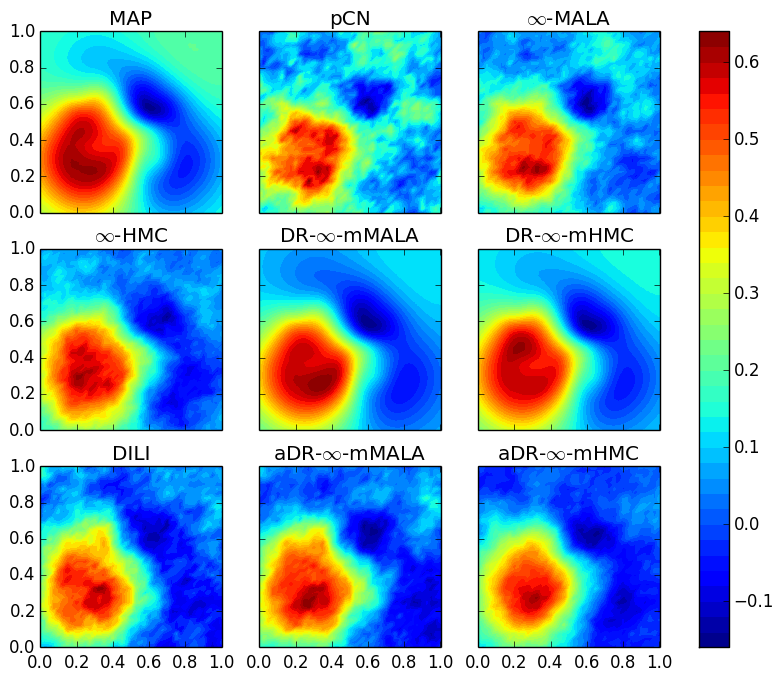}
\caption{Elliptic inverse problem ($\textrm{SNR}=10$): Bayesian posterior mean estimates of the log-transmissivity field $u(s)$ based on $2000$ samples by various MCMC algorithms; the upper-left corner shows the MAP estimate.}
\label{fig:est_snr10}
\end{center}
\end{figure}

\begin{table}[ht]
\centering
\begin{tabular}{l|ccccccc}
  \hline
Method & h & AP & s/iter & ESS(min,med,max) & minESS/s & spdup & PDEsolns \\ 
  \hline
pCN & 0.50 & 0.74 & 0.79 & (52.91,171.86,283.42) & 0.0336 & 1.00 & 2501 \\ 
  $\infty$-MALA & 2.00 & 0.75 & 1.53 & (223.26,644.39,912.85) & 0.0732 & 2.18 & 5002 \\ 
  $\infty$-HMC & 1.30 & 0.71 & 3.68 & (897.11,1469.7,2000) & 0.1219 & 3.63 & 12342 \\ 
  DR-$\infty$-mMALA & 6.00 & 0.72 & 9.24 & (696.83,1040.67,2000) & 0.0377 & 1.12 & 80032 \\ 
  DR-$\infty$-mHMC & 4.00 & 0.78 & 22.47 & (887.24,1212.38,1699.42) & 0.0197 & 0.59 & 198176 \\ 
  DILI & (0.5,\,1.0) & 0.70 & 1.60 & (214.58,580.81,807.74) & 0.0670 & 1.99 & 5806 \\ 
  aDR-$\infty$-mMALA & 3.00 & 0.75 & 1.55 & (336.75,887.82,1180.41) & 0.1084 & {\bf 3.23} & 5806 \\ 
  aDR-$\infty$-mHMC & 1.50 & 0.77 & 3.80 & (1053.06,1956.4,2000) & 0.1387 & {\bf 4.13} & 13370 \\ 
   \hline
\end{tabular}
\caption{Sampling efficiency in elliptic inverse problem (SNR=10). Column labels are as follows.
h: step size(s) used for making MCMC proposal;
AP: average acceptance probability; s/iter: average seconds per iteration; ESS(min,med,max): minimum, median, maximum of Effective Sample Size across all posterior coordinates; min(ESS)/s: minimum ESS per second;
spdup: speed-up relative to base pCN algorithm;
PDEsolns: number of PDE solutions during execution.} 
\label{tab:elliptic_snr10}
\end{table}

The transmissivity field is endowed with a log-Gaussian prior, i.e.
\begin{equation*}
k(s) = \exp(u(s)), \quad u(s) \sim \mathcal N(0, \Co)
\end{equation*}
where the covariance operator $\Co$ is defined through an exponential kernel function
\begin{equation*}
\Co: \Hi \rightarrow \Hi, \; u(s) \mapsto \int c(s, s') u(s') ds', \quad c(s, s') = \sigma_u^2 \exp \left( -\frac{\Vert s- s' \Vert}{2s_0}\right), \,\textrm{for}\; s,s' \in \Omega
\end{equation*}
with the prior standard deviation $\sigma_u=1.25$ and the correlation length $s_0=0.0625$ in the experiments.
To make the inverse problem more challenging, we follow \citep{cui16} to use a true transmissivity field $k_0(s)$ that is not drawn from the prior,
as shown on the left panel of Figure \ref{fig:truth_obs}.
The right panel of Figure \ref{fig:force_soln} shows the potential function $p(s)$, solved
with $k_0(s)$, which is also used for generating noisy observations.
Partial observations are obtained by solving $p(s)$ on an $80\times 80$ mesh and then collecting at $25$ measurement sensors as shown by the circles on the right panel of Figure \ref{fig:truth_obs}. 
The corresponding observation operator $\mathcal O$ yields the data
\begin{equation*}
y = \mathcal O p(s) + \eta, \quad \eta \sim \mathcal N(0, \sigma_\eta^2 I_{25})
\end{equation*}
Define the signal-to-noise ratio (SNR) as $\max_s \{u(s)\}/\sigma_\eta$.
We consider $\textrm{SNR}=10$ and $\textrm{SNR}=100$, data in the latter case containing more information than those in the former case.

The inverse problem involves sampling from the posterior of the log-transmissivity field $u(s)$, which becomes a vector of dimension over $6000$ after being discretized on $40\times 40$ mesh (with Lagrange degree $2$).
We compare the performance of algorithms including pCN, $\infty$-MALA, $\infty$-HMC, DR-$\infty$-mMALA, DR-$\infty$-mHMC, DILI, aDR-$\infty$-mMALA and aDR-$\infty$-mHMC \footnote{Implementation of full $\infty$-mMALA and $\infty$-mHMC is too intensive and therefore omitted.}. For each algorithm, we run $2500$ iterations and burn in the first $500$. For HMC algorithms, we let $I=4$.
For the location-dependent algorithms, we fix the dimension of local intrinsic subspaces at $r=5$. And for adaptive algorithms, we set the interval for updating global LIS $n_{\textrm{lag}}=200$,  
the interval for updating projected covariance $n_b=50$ (for DILI), and stop when either the F\"orstner diagnostic \citep{cui16} falls below the threshold $\Delta_{\textrm{LIS}}=10^{-5}$ or it reaches the maximum step $m_{\max}=100$. We tune the step sizes for each algorithm so that they have similar acceptance rate $60 \sim 70 \%$.

\begin{figure}[t]
\begin{center}
\includegraphics[width=1\textwidth,height=.4\textwidth]{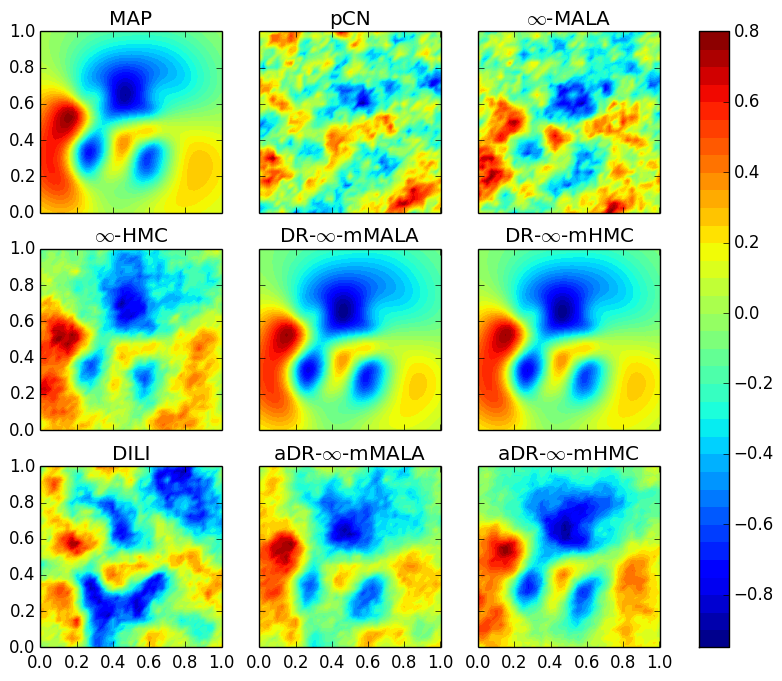}
\caption{Elliptic inverse problem ($\textrm{SNR}=100$): Bayesian posterior mean estimates of the log-transmissivity field $u(s)$ based on $2000$ samples by various MCMC algorithms; the upper-left corner shows the MAP estimate.}
\label{fig:est_snr100}
\end{center}
\end{figure}

\begin{table}[htbp]
\centering
\begin{tabular}{l|ccccccc}
  \hline
Method & h & AP & s/iter & ESS(min,med,max) & minESS/s & spdup & PDEsolns \\ 
  \hline
pCN & 0.01 & 0.57 & 0.99 & (2.67,6.95,37.79) & 0.0013 & 1.00 & 2501 \\ 
  $\infty$-MALA & 0.04 & 0.61 & 1.62 & (4.32,15.34,51.45) & 0.0013 & 0.99 & 5002 \\ 
  $\infty$-HMC & 0.04 & 0.59 & 3.52 & (24.36,92.13,184.84) & 0.0035 & 2.57 & 12342 \\ 
  DR-$\infty$-mMALA & 0.52 & 0.67 & 8.85 & (127.25,210.84,460.07) & 0.0072 & 5.34 & 80032 \\ 
  DR-$\infty$-mHMC & 0.25 & 0.56 & 22.97 & (190.2,322.29,687.11) & 0.0041 & 3.08 & 198176 \\ 
  DILI & (0.1, 0.2) & 0.69 & 1.59 & (30.52,133.67,221.97) & 0.0096 & {\bf 7.13} & 6612 \\ 
  aDR-$\infty$-mMALA & 0.25 & 0.71 & 1.61 & (12.09,89.17,174.36) & 0.0037 & 2.79 & 6612 \\ 
  aDR-$\infty$-mHMC & 0.10 & 0.69 & 3.63 & (70.99,234.42,364.31) & 0.0098 & {\bf 7.26} & 14056 \\ 
   \hline
\end{tabular}
\caption{Sampling efficiency in the elliptic inverse problem (SNR=100). Column labels are as follows.
h: step size(s) used for making MCMC proposal;
AP: average acceptance probability; s/iter: average seconds per iteration; ESS(min,med,max): minimum, median, maximum of Effective Sample Size across all posterior coordinates; min(ESS)/s: minimum ESS per second;
spdup: speed-up relative to base pCN algorithm;
PDEsolns: number of PDE solutions during execution.} 
\label{tab:elliptic_snr100}
\end{table}

We first present the results for $\textrm{SNR}=10$.
Figure \ref{fig:est_snr10} shows the mean estimates of the log-transmissivity field $u(s)$ based on $2000$ posterior samples.
All MCMC algorithms generate estimates consistent with the maximum a posterior (MAP).
However, pCN gives the noisiest estimate because of the high auto-correlation in its samples;
two location-dependent algorithms, DR-$\infty$-mMALA and DR-$\infty$-mHMC, though taking longer time to finish, output the best results closest to the MAP estimate.
Note that among the three globally adaptive algorithms, DILI, aDR-$\infty$-mMALA and aDR-$\infty$-mHMC, they have similar estimates, with the last one being the best,
due to the superiority of multiple steps that can suppress the random walk behavior.

Quantitatively, their sampling efficiency, mainly measured by minimum effective sample size (ESS) per unit time, is compared in Table \ref{tab:elliptic_snr10}.
In general, HMC algorithms have larger ESS compared to their MALA  analogies, demonstrating the advantage of multi-step proposals.
The efficiency gain by local pre-conditioners in DR-$\infty$-mMALA and DR-$\infty$-mHMC, is undermined by the extra computational burden of low-rank approximation done for every MCMC iteration.
The globally adaptive versions, aDR-$\infty$-mMALA and aDR-$\infty$-mHMC, have higher ESS compared to other manifold algorithms, but with lower computational cost roughly equivalent to that of non-manifold algorithms (e.g. $\infty$-MALA and $\infty$-HMC respectively), thus achieve the highest sampling efficiency.
They are shown to be more efficient than DILI.

\begin{figure}[htbp]
  \begin{center}
     \includegraphics[width=1\textwidth,height=.35\textwidth]{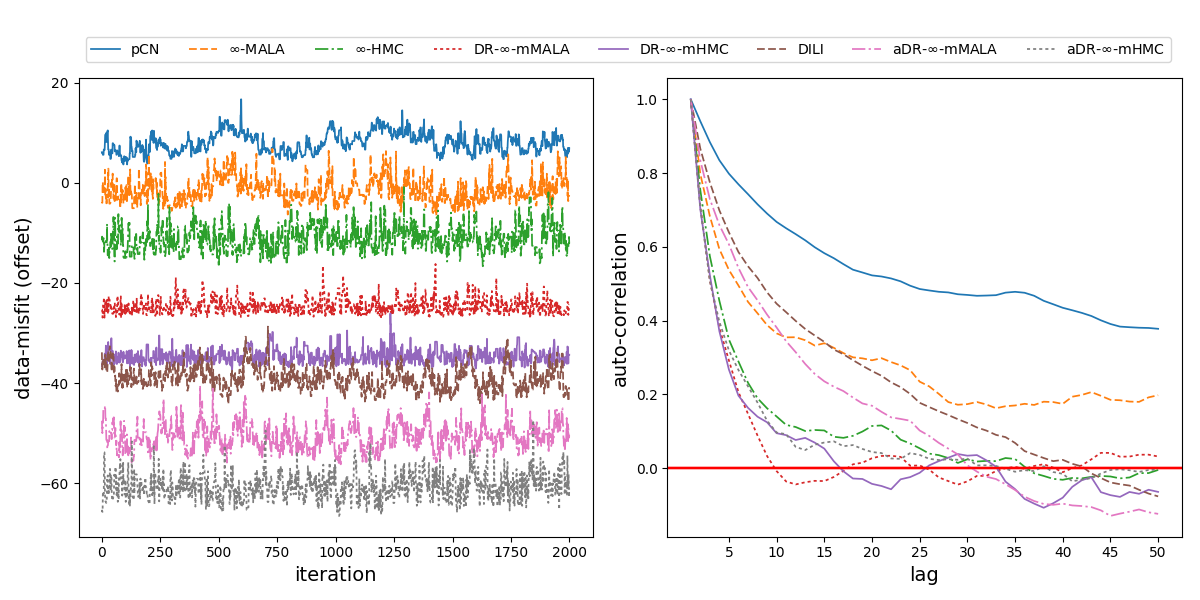}
  \end{center}
  \caption{Elliptic inverse problem ($\textrm{SNR}=100$): the trace plots of data-misfit function evaluated with each sample (left, values have been offset to be better compared with) and the auto-correlation of data-misfits as a function of lag (right).}
  \label{fig:misfit_acf_elliptic_snr100}
\end{figure}

\begin{figure}[t]
  \begin{center}
     \includegraphics[width=1\textwidth,height=.4\textwidth]{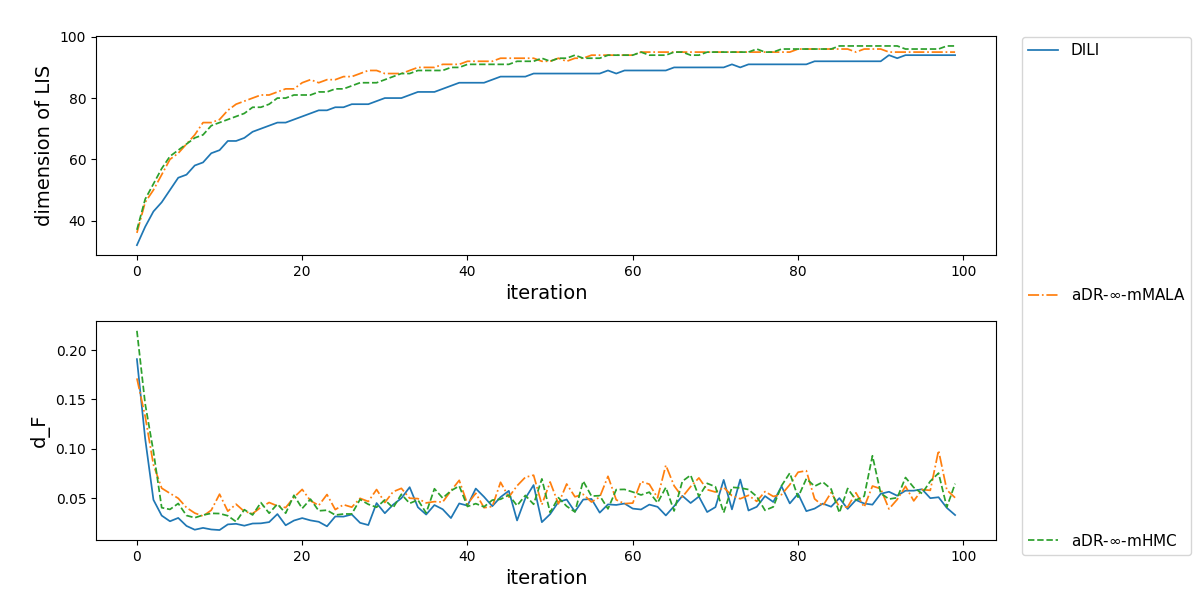}
  \end{center}
  \caption{The adaptation of the intrinsic subspace (LIS) in the elliptic inverse problem ($\textrm{SNR}=100$).
  Upper: growing dimensions of the global Likelihood Informed Subspace (LIS); Lower: evolving F\"orstner diagnostics $d_{\mathcal F}$.}
  \label{fig:LIS_updates_elliptic_snr100}
\end{figure}

Now we consider the case $\textrm{SNR}=100$.
In this case, the data contain more information since they are contaminated by noise of smaller magnitude.
Therefore, we set $n_{\textrm{lag}}=100$ in this case.
Figure \ref{fig:est_snr100} shows the mean estimates of the log-transmissivity field $u(s)$ based on $2000$ posterior samples.
The landscape of $u(s)$ is more complicated than the previous case. 
Again, two location-dependent algorithms, DR-$\infty$-mMALA and DR-$\infty$-mHMC, generate samples with the highest quality that yield results closest to the MAP estimate.
Note, estimates by both aDR-$\infty$-mMALA and aDR-$\infty$-mHMC are significantly better than that by DILI.

Table \ref{tab:elliptic_snr100} summarizes the sampling efficiency of all these algorithms.
Note in this case, aDR-$\infty$-mMALA's lower sampling efficiency compared to DILI may indicate that the intrinsic global LIS has rich geometric information that diagonal approximation to the projected $\Ko(v)$ as in \eqref{eq:lo-postcov} is not sufficient to support effective exploration;
however such insufficiency is remedied by multi-step transition movement such that aDR-$\infty$-mHMC could outperform DILI.

Figure \ref{fig:misfit_acf_elliptic_snr100} further illustrates the quality of samples from all the MCMC algorithms.
The left panel shows their data-misfit values \eqref{eq:gauss_nz}, which have been offset in order to better contrast their difference.
PCN is seen to have the most sticky trace-plot, indicating the highest auto-correlation in data-misfits; while those with large ESS have non-adhesive samples. This is verified by the low auto-correlations on the right panel. 
PCN has the highest auto-correlation at all lags, followed by $\infty$-MALA, aDR-$\infty$-mMALA and DILI with similar values.
In general `HMC' type algorithms have lower auto-correlation values compared with their `MALA' type counterparts.

\begin{figure}[t]
  \begin{center}
     \includegraphics[width=1\textwidth,height=.4\textwidth]{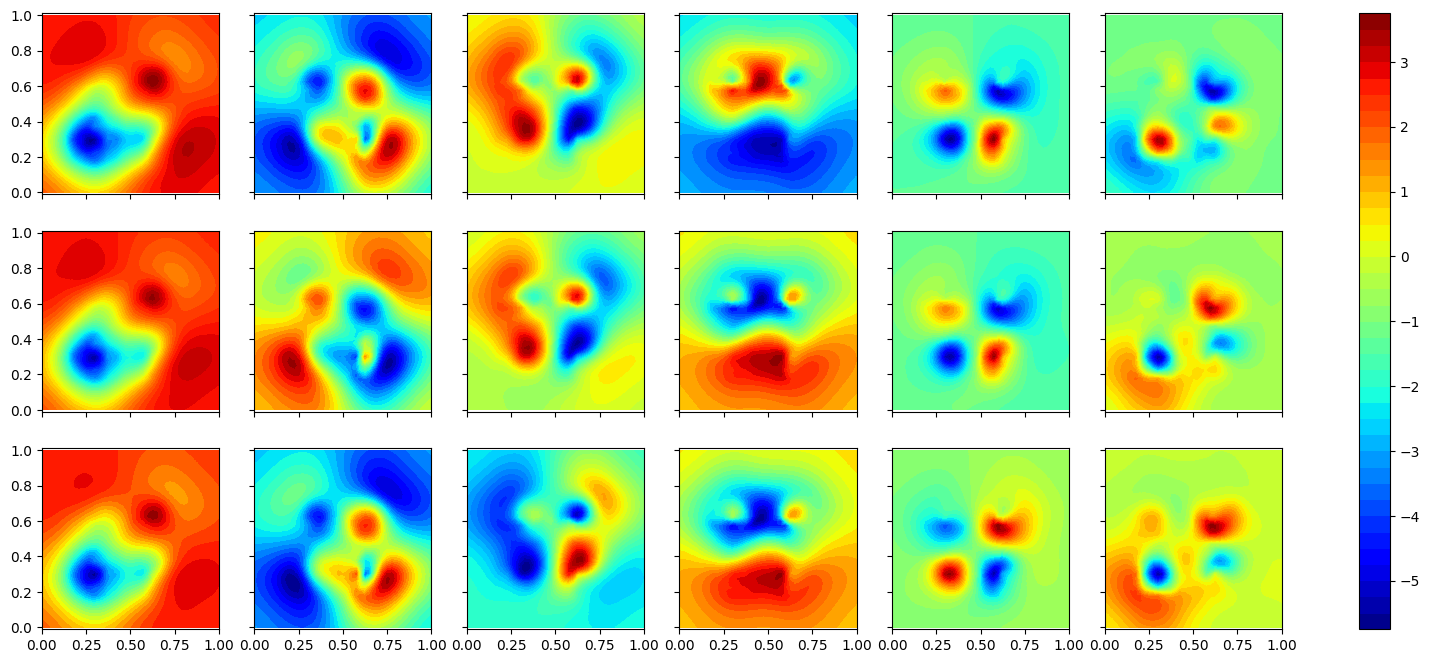}
  \end{center}
  \caption{Principal eigenfunctions projected to the global LIS in the elliptic inverse problem ($\textrm{SNR}=100$).
  Top: DILI; Middle: aDR-$\infty$-mMALA; Bottom: aDR-$\infty$-mHMC.}
  \label{fig:peigfs_updates_elliptic_snr100}
\end{figure}

To investigate the adaptation process of DILI, aDR-$\infty$-mMALA and aDR-$\infty$-mHMC, we run each for $10^4$ iterations and update the global LIS every $n_\textrm{lag}=100$ iterations.
Figure \ref{fig:LIS_updates_elliptic_snr100} shows the adaptation of the global LIS as a function of updating iteration.
They all have the similar trajectories of growing dimensions and F\"orstner diagnostics.
The final dimensions of the global LIS for DILI, aDR-$\infty$-mMALA and aDR-$\infty$-mHMC are $94, 95, 97$ respectively.
Figure \ref{fig:peigfs_updates_elliptic_snr100} demonstrates 6 principal eigen-functions of the global LIS identified by these adaptive algorithms and they are comparable to each other despite of a negative sign.

\subsection{RANS inverse problem}\label{sec:rans}

Now we study model inadequacy in the context of $k$-$\epsilon$ Reynolds-Averaged Navier-Stokes (RANS) model of turbulent jet.
RANS equations are time-averaged equations of motion for fluid flow, primarily used to describe turbulent flows \citep{reynolds1895,TENNEKES1992}.
Based on knowledge of the properties of flow turbulence, these equations can be used to give approximate time-averaged solutions to the Navier-Stokes equations.
We start by defining the problem. Consider the following incompressible Navier-Stokes equation with constant density:
\begin{equation}\label{eq:NS}
\begin{aligned}
\textrm{Continuity}:\quad &
\frac{\partial u_i}{\partial x_i} = 0\ , \\
\textrm{Momentum}:\quad &
\frac{\partial  u_i}{\partial t} + u_j\frac{\partial u_i}{\partial x_j} = -\frac{\partial p}{\partial x_i} + \nu \frac{\partial^2 u_i}{\partial x_j \partial x_j}\ ,  
\end{aligned}
\end{equation}
We take the average of the above system (denoting the average value as $U := \bar{u}$ and the residual as $u' = u - U$),
and close it by modeling the \emph{Reynolds-stress tensor}, $-\overline{u'_i u'_j}$
by means of the commonly used $k$-$\epsilon$ model.
This leads to the $k$-$\epsilon$ RANS equations as follows:
\begin{equation}
  \begin{aligned}
    \frac{\partial U_i}{\partial x_i} &= 0  \\
  \frac{\partial  U_i}{\partial t} 
  &= - U_j\frac{\partial U_i }{\partial x_j} -\frac{\partial P}{\partial x_i} + \frac{\partial}{\partial x_j}
  \left[ \nu \frac{\partial U_i}{\partial x_j} +C_\mu \frac{k^2}{\epsilon}\left(\frac{\partial U_i }{\partial x_j}
  +\frac{\partial U_j}{\partial x_i} \right) - \frac{2}{3}k\delta_{ij} \right] \\
  \frac{\partial k }{\partial t}   &= - U_j \frac{\partial k}{\partial x_j}
 + \left(    C_\mu \frac{k^2}{\epsilon}\left(\frac{\partial U_i }{\partial x_j}
  +\frac{\partial U_j}{\partial x_i} \right) - \frac{2}{3}k\delta_{ij} \right)\frac{\partial U_i}{\partial x_j} - \epsilon
  + \frac{\partial}{\partial x_j}
  \left[  \left(\nu + \frac{C_\mu k^2}{\sigma_k \epsilon}\right)\frac{\partial k}{\partial x_j}\right] \\
  \frac{\partial \epsilon }{\partial t}
    &= - U_j\frac{\partial \epsilon}{\partial x_j} + C_{\epsilon 1} \frac{\epsilon}{k} \left(    C_\mu \frac{k^2}{\epsilon}\left(\frac{\partial U_i }{\partial x_j}
  +\frac{\partial U_j}{\partial x_i} \right) - \frac{2}{3}k\delta_{ij} \right)\frac{\partial U_i}{\partial x_j}
  - C_{\epsilon 2}\frac{\epsilon^2}{k} + \frac{\partial}{\partial x_j}
  \left[  \left(\nu + \frac{C_\mu k^2}{\sigma_\epsilon \epsilon}\right) \frac{\partial \epsilon}{\partial x_j}\right]
  \end{aligned}
  \label{eq:ke_RANS}
\end{equation}
where $U_i$ and $P$ denote the averaged velocity and pressure; $k$
and $\varepsilon$ denote the turbulent kinetic energy and turbulent
dissipation; $\nu$ and $\nu_t = C_{\mu} \frac{k^2}{\varepsilon}$ are
the kinematic and turbulent viscosity; and $C_\mu$, $\sigma_k$,
$\sigma_\epsilon$, $C_{\epsilon 1}$, $C_{\epsilon 2}$ are the
empirical constants involved in the $k$-$\varepsilon$ closure
model. These PDEs are augmented by appropriate inflow and outflow
boundary conditions detailed in \ref{apx:DNS}. These conditions are imposed
with mollified operators to enforce positivity of
($k,\varepsilon$) and switch to a Dirichlet boundary condition when an
inflow is detected on the outflow boundary.

\begin{figure}[t]
  \centering
  \includegraphics[width=\textwidth,height=.35\textwidth]{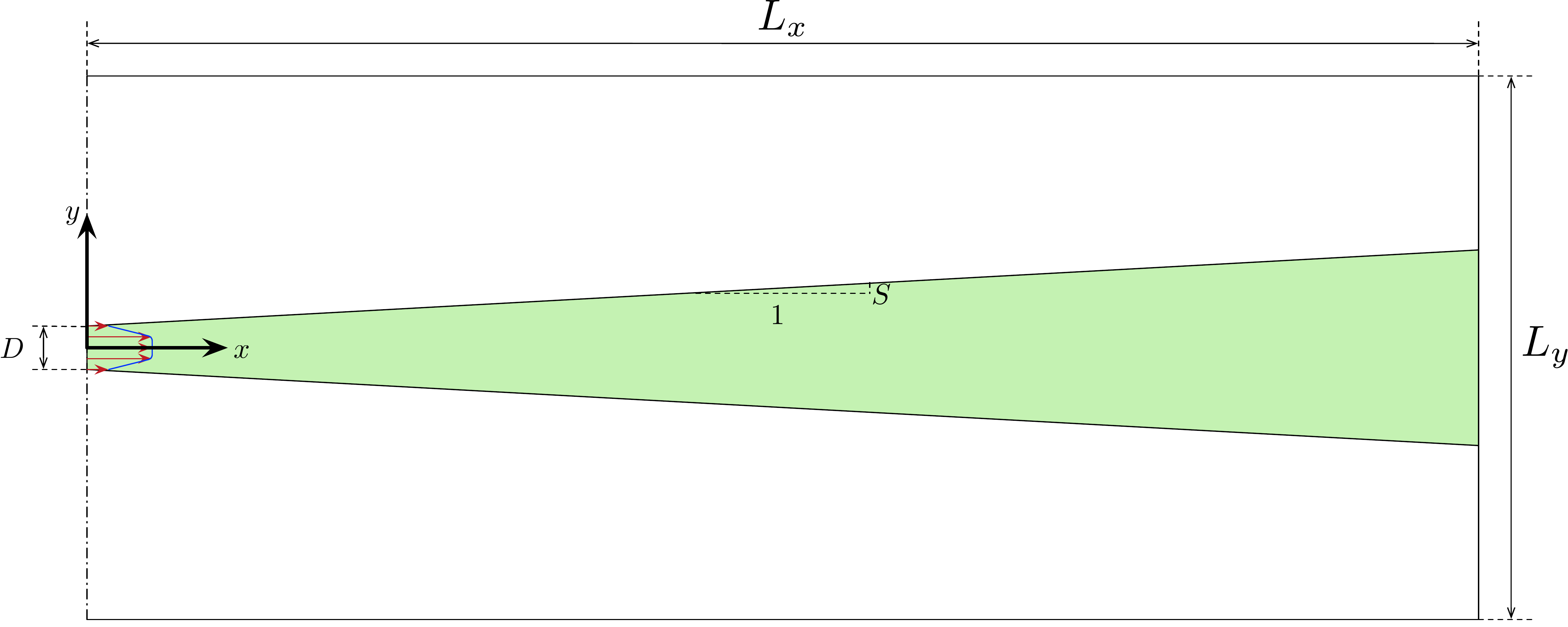}
  \caption{Geometry of non-reacting 2D jet flow simulation}
  \label{fig:2d_jet_flow}
\end{figure}

The $k$-$\epsilon$ RANS model \eqref{eq:ke_RANS} serves as a cheaper approximation
to the Navier-Stokes equation \eqref{eq:NS} whose high resolution solutions by direct 
numerical simulation (DNS) may be prohibitively expensive to obtain \citep{klein03}.
However the model may be inherently unable to characterize certain phenomena present
in the fully resolved Navier-Stokes equation. To understand how and when this
occurs, model inadequacy is represented by replacing the constant $C_\mu$ with 
the field $e^m C_\mu$ ($m$ an unknown field to be determined). We take a Bayesian
approach to calibrate $m$ to high resolution DNS data obtained in \cite{klein03}.
Figure \ref{fig:2d_jet_flow} shows the simulation domain and Figure \ref{fig:DNS_VIS} 
shows the results from DNS data.
One can refer \ref{apx:DNS} for more simulation details.

We focus on the upper half domain in Figure \ref{fig:2d_jet_flow}.
A damped Newton method is employed to solve the nonlinear system of equations \eqref{eq:ke_RANS} on a $40\times 80$ finite element mesh till its steady-state.
Pseudo-time continuation is used to guarantee the global convergence to a physically stable solution.
The uncertainty field $m$ is represented a priori by a Gaussian random field, 
and the likelihood function of $m$ is derived by fitting the DNS data to the forward 
solution of the $k$-$\epsilon$ RANS equation \eqref{eq:ke_RANS} with $C_\mu$ replaced by $e^m C_\mu$.
The inverse problem involves sampling the posterior of uncertainty field $m$ given
DNS data.
The field $m$ becomes a parameter vector of over $3000$ dimensions after discretization. Therefore the inference about it is computationally challenging.

\begin{figure}[t]
\begin{center}
\includegraphics[width=1\textwidth,height=.4\textwidth]{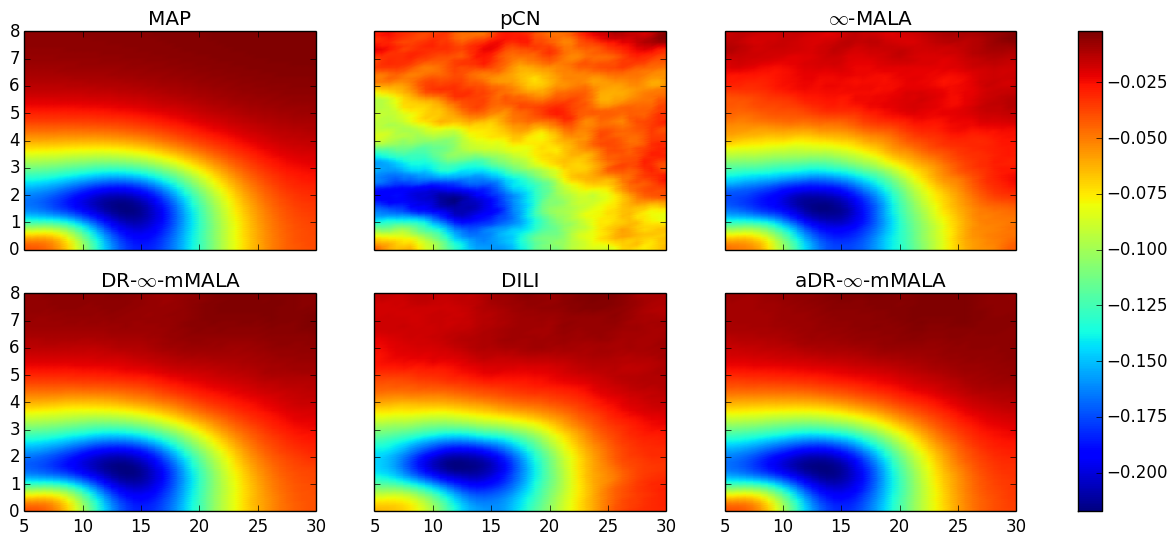}
\caption{Inverse RANS problem: Bayesian posterior mean estimates of the log-uncertainty field $m(x)$ based on $2000$ samples by various MCMC algorithms; the upper-left corner shows the MAP estimate.}
\label{fig:est_rans_l25w8_2000}
\end{center}
\end{figure}

\texttt{FEniCS} \citep{LoggMardalEtAl2012a,AlnaesBlechta2015a} codes are developed and \texttt{hIPPYlib} \citep{villa2018} library is adopted by EQUiPS team led by the University of Texas, Austin. 
Those codes are adapted to run MCMC algorithms including pCN, $\infty$-MALA, DR-$\infty$-mMALA, DILI and aDR-$\infty$-mMALA. Due to the extremely intensive computation required for repeated forward solving, we only focus on `MALA' type algorithms in this example.
Figure \ref{fig:est_rans_l25w8_2000} shows the point estimates including MAP based on optimization and posterior means of $2000$ MCMC samples after burning the first $500$.
All the MCMC results are consistent with MAP, though pCN and $\infty$-MALA give more noisy estimates than others.
Table \ref{tab:rans_l25w8_2000} compares the sampling efficiency, measured by minimal ESS per second, among these five algorithms.
Note, our proposed dimension reduced algorithms, DR-$\infty$-mMALA and aDR-$\infty$-mMALA, have achieved much more speed up compared to DILI.
These results demonstrate the benefit of geometric information and success of dimension reduction in accelerating MCMC in this highly challenging model inadequacy problem.
Based on the estimates, we see that there is considerably more model error
($m$ far from zero) in the tail region (between $10$ and $20$ in $x$) of the flame, when using the $k$-$\epsilon$ RANS model \eqref{eq:ke_RANS} to approximate the Navier-Stokes equation \eqref{eq:NS}.


\begin{table}[ht]
\centering
\begin{tabular}{l|ccccccc}
  \hline
Method & h & AP & s/iter & ESS(min,med,max) & minESS/s & spdup & PDEsolns \\ 
  \hline
pCN & 0.01 & 0.75 & 31.32 & (3.07,9.01,34.66) & 4.91e-05 & 1.00 & 2501 \\ 
  $\infty$-MALA & 0.32 & 0.72 & 70.17 & (35.24,131.51,360.09) & 2.51e-04 & 5.12 & 5002 \\ 
  DR-$\infty$-mMALA & 28.66 & 0.69 & 163.51 & (780.77,1934.9,2000) & 2.39e-03 & {\bf 48.67} & 305122 \\ 
  DILI & (0.24,0.40) & 0.77 & 79.32 & (73.75,288.54,411.71) & 4.65e-04 & 9.48 & 6664 \\ 
  aDR-$\infty$-mMALA & 6.00 & 0.87 & 123.69 & (941.78,2000,2000) & 3.81e-03 & {\bf 77.61} & 6664 \\ 
   \hline
\end{tabular}
\caption{Sampling efficiency in the inverse RANS problem. Column labels are as follows.
h: step size(s) used for making MCMC proposal;
AP: average acceptance probability; s/iter: average seconds per iteration; ESS(min,med,max): minimum, median, maximum of Effective Sample Size across all posterior coordinates; min(ESS)/s: minimum ESS per second;
spdup: speed-up relative to base pCN algorithm;
PDEsolns: number of PDE solutions during execution.} 
\label{tab:rans_l25w8_2000}
\end{table}

\begin{figure}[t]
  \begin{center}
     \includegraphics[width=1\textwidth,height=.35\textwidth]{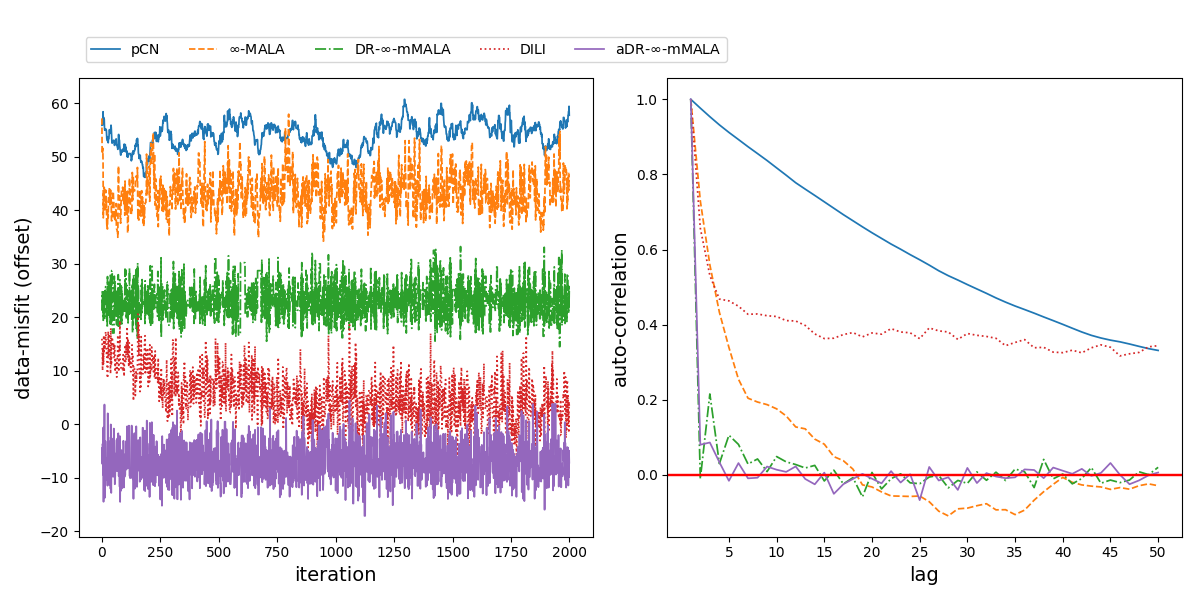}
  \end{center}
  \caption{Inverse RANS problem: the trace plots of data-misfit function evaluated with each sample (left, values have been offset to be better compared with) and the auto-correlation of data-misfits as a function of lag (right).}
  \label{fig:misfit_acf_rans_l25w8_2000}
\end{figure}

Figure \ref{fig:misfit_acf_rans_l25w8_2000} further illustrates the quality of samples from all these MCMC algorithms.
The left panel shows the offset data-misfit function \eqref{eq:gauss_nz} evaluations, and the right panel plots their auto-correlation functions of lag.
The results can be divided into three groups: pCN is the worst with sticky trace-plot and highest auto-correlation; then DILI behaves similarly as $\infty$-MALA, though they differ in auto-correlation after lag 5; 
while DR-$\infty$-mMALA and aDR-$\infty$-mMALA perform the best in efficiently producing posterior samples with small auto-correlations.

\begin{figure}[htbp]
  \begin{center}
     \includegraphics[width=1\textwidth,height=.5\textwidth]{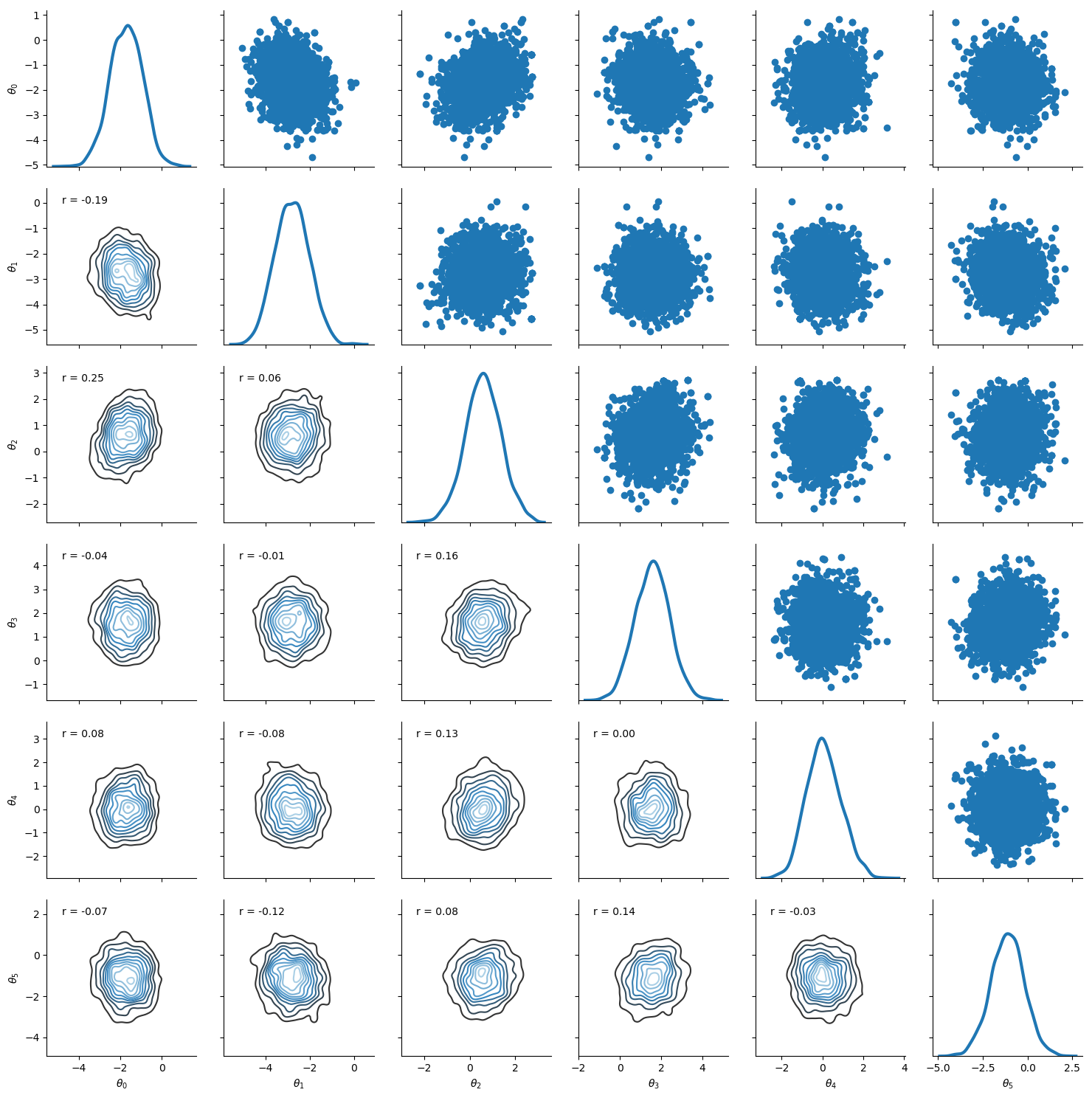}
  \end{center}
  \caption{Inverse RANS problem: the pairwise density contours of $\Vo^*_6 v$ in the projected global LIS $\Hi_6$. $(\theta_1,\cdots,\theta_6):=\Vo^*_6 v$.}
  \label{fig:pdc_rans_l25w8_2000}
\end{figure}

As mentioned in the previous context, the intrinsic subspace has rich geometric information concentrated on low-dimensions.
Figure \ref{fig:pdc_rans_l25w8_2000} plots the pairwise posterior density contours of $\Vo^*_6 v$ in the projected global LIS $\Hi_6$ of 6 principal dimensions.
One can observe that they substantially deviate from Gaussian distributions.


\section{Conclusion and Discussion}\label{sec:conclusion}

In this paper, we accelerate $\infty$-GMC \citep{beskos2017} with dimension reduction.
The dimension reduction techniques we propose are directly based on the low-rank approximation of the (prior or Gaussian-approximate posterior) covariance operator.
Randomized algorithms \citep{halko11,saibaba16,liberty07} are used to obtain principal eigen-directions, which form the basis of an intrinsic low-dimensional subspace.
Geometry-informed algorithms are applied to the intrinsic subspace to effectively probe its complex structure; while the geometry-flat complementary subspace can be efficiently explored using simpler methods.
We develop location-dependent implementation (DR-$\infty$-mMALA and DR-$\infty$-mHMC) and globally adaptive implementation (aDR-$\infty$-mMALA and aDR-$\infty$-mHMC).
ADR-$\infty$-mMALA is analogous to DILI \citep{cui16}, but different in the low-rank approximation as well as the acceptance probability.
Interesting connections are established between them but our proposed methods are shown to be simpler and more efficient. ADR-$\infty$-mHMC goes beyond DILI to make multi-step proposals that can suppress the random walk behavior.
We also compare various dimension-independent MCMC proposals and quantify their difference with upper bounds. They are presented to predict their asymptotic behavior as the dimension of intrinsic subspace increases.
Numerical evidence, including an elliptic inverse problem and an inverse problem involving RANS equations, supports the computationally advantage over the state-of-the-art counterparts.

One interesting future direction could be a combination of sequential Monte Carlo (SMC) \citep{del2006sequential} and the proposed dimension reduced $\infty$-GMC (DR-$\infty$-GMC) to parallelize MCMC.
SMC samplers deal with sequentially sampling from 
the distributions $\{\pi_n\}_{n\in\mathbb T}$ that are defined on a common measurable space $(E, \mathcal E)$.
SMC builds on the sequential importance sampling (SIS), 
which moves the particles $X_{n-1}^{(i)} \sim \eta_{n-1}$ according to the Markov kernel
$K_n: E\times \mathcal E\rightarrow [0,1]$:
\begin{equation*}
X_n^{(i)} \sim \eta_n(x') = \int_E \eta_{n-1} (x) K_n(x,x') dx \ .
\end{equation*}
\cite{beskos2015sequential} uses pCN for $K_n$ in the SIS framework and has demonstrated orders of magnitude speed up of pCN by combining it with SMC in this way.
Substituting $K_n$ with DR-$\infty$-GMC Markov kernel into the SMC scheme can parallelize these MCMC algorithms
to achieve further efficiency improvement to tackle Bayesian inverse problems at larger scale.

DR-$\infty$-mHMC can be improved further by
e.g., surrogate methods \citep{lan16,zhang2017b} or grid methods \citep{zhang2017a} to reduce the burden of point-wise updating gradient and metric.
Within the leap-frog steps of `HMC' type algorithms, one can consider `BFGS' type update as in quasi-Newton methods \citep{zhang11}.

The methods proposed in this paper work well for inverse problems with Gaussian priors. 
Recently, there is a growing interest in flexible modeling with non-Gaussian priors \citep{chen2018,hosseini2018,zahm2018}, including Besov priors, level-set priors, deep Gaussian priors, Bessel-K prior, etc..
They are employed to induce sparse MAP estimators, to do graph-based classification, or to represent non-Gaussian phenomena.
\cite{chen2018} represent typical non-Gaussian priors as a hierarchy of conditional Gaussian priors using whitening transformations and convert problems to Gaussian prior based sampling (pCN).
It will be useful 
to develop geometry-aided dimension robust algorithms for inverse problems with non-Gaussian priors.

\section*{Acknowledgement}
SL was supported by the DARPA funded program Enabling Quantification of Uncertainty in Physical Systems (EQUiPS), contract W911NF-15-2-0121 when the paper was written.
We thank EQUiPS team for sharing the FEniCS codes for solving $k$-$\eps$ RANS equations (for both forward and adjoint problems), and especially Umberto Villa at the University of Texas-Austin for numerous help.
We also thank the anonymous reviewers for the constructive comments that help to improve the manuscript.

\def\urlprefix{}
\def\url#1{\vspace{-1em}}
\bibliography{references}

\newpage
\begin{center}
{\Large \bf Appendix: PDE setting and Proofs}
\end{center}
\appendix

\section{Direct Numerical Simulation details of RANS \citep[chap. 3 of][]{equips2016}}\label{apx:DNS}

\subsection{$k$-$\epsilon$ RANS equations}
Taking averages $\bar{\cdot}$ on  \eqref{eq:NS} gives
\begin{equation}
    \frac{\partial U_i}{\partial x_i} = 0 
  \label{eq:avg_continuity}
\end{equation}
\begin{equation}
  \frac{\partial  U_i}{\partial t} + U_j\frac{\partial U_i }{\partial x_j}
  = -\frac{\partial P}{\partial x_i} + \frac{\partial}{\partial x_j}
  \left( \nu \frac{\partial U_i}{\partial x_j} - \overline{u'_i u'_j} \right)
  \label{eq:avg_momentum}
\end{equation}
where Equation \eqref{eq:avg_momentum} is referred to as the
Reynolds-Averaged Navier-Stokes (RANS) equation.
We close the above system by modeling the \emph{Reynolds-stress tensor}, $-\overline{u'_i u'_j}$
by means of the $k$-$\epsilon$ model. 
\begin{equation}
  \nu_t = C_\mu \frac{k^2}{\epsilon}
  \label{eq:def_nu_t}
\end{equation}
\begin{equation}
  \mean{u'_i u'_j} = -\nu_t\left(\frac{\partial U_i }{\partial x_j}
  +\frac{\partial U_j}{\partial x_i} \right) + \frac{2}{3}k\delta_{ij}
  \label{eq:model_rey_stress}
\end{equation}
\begin{equation}
  \frac{\partial k }{\partial t} + U_j \frac{\partial k}{\partial x_j}
  = - \mean{ u'_i u'_j}\frac{\partial U_i}{\partial x_j} - \epsilon
  + \frac{\partial}{\partial x_j}
  \left[  \left(\nu + \frac{\nu_t}{\sigma_k}\right)\frac{\partial k}{\partial x_j}\right]
  \label{eq:k_eq}
\end{equation}
\begin{equation}
  \frac{\partial \epsilon }{\partial t}
  + U_j\frac{\partial \epsilon}{\partial x_j}
  = - C_{\epsilon 1} \frac{\epsilon}{k} \mean{u'_i u'_j}\frac{\partial U_i}{\partial x_j}
  - C_{\epsilon 2}\frac{\epsilon^2}{k} + \frac{\partial}{\partial x_j}
  \left[  \left(\nu + \frac{\nu_t}{\sigma_\epsilon}\right) \frac{\partial \epsilon}{\partial x_j}\right]
  \label{eq:e_eq}
\end{equation}
where the commonly used coefficients are 
\begin{equation*}
    C_\mu = 0.09, \quad
    \sigma_k = 1.00, \quad
    \sigma_\epsilon = 1.30, \quad
    C_{\epsilon 1} = 1.44, \quad
    C_{\epsilon 2} = 1.92 
\end{equation*}

Substituting \eqref{eq:def_nu_t} and \eqref{eq:model_rey_stress} into
\eqref{eq:avg_momentum}, \eqref{eq:k_eq} and \eqref{eq:e_eq} leads to the $k$-$\epsilon$ RANS equation \eqref{eq:ke_RANS}.
The following convention is used when there is no confusion:
\begin{equation*}
u=(u_1,u_2)=(u,v); \qquad U=\mean{u}, V=\mean{v}=\mean{u_2}
\end{equation*}

\subsection{Geometry}
The geometry of two dimensional non-reacting jet flow is shown Figure \ref{fig:2d_jet_flow}.
The simulation domain size is chosen to be $L_x = 20D$ and $L_y = 8D$ where $D$ is nozzle width of inlet flow.
The mean velocity profile and turbulent kinetic energy from DNS data is shown in Figure \ref{fig:DNS_VIS}.
We here follow simulation geometry and boundary conditions of the work by \cite{klein03}.

\begin{figure}[t]
    \hspace{-15pt}
    \begin{subfigure}[b]{0.32\textwidth}
        \includegraphics[width=1.2\textwidth,height=1\textwidth]{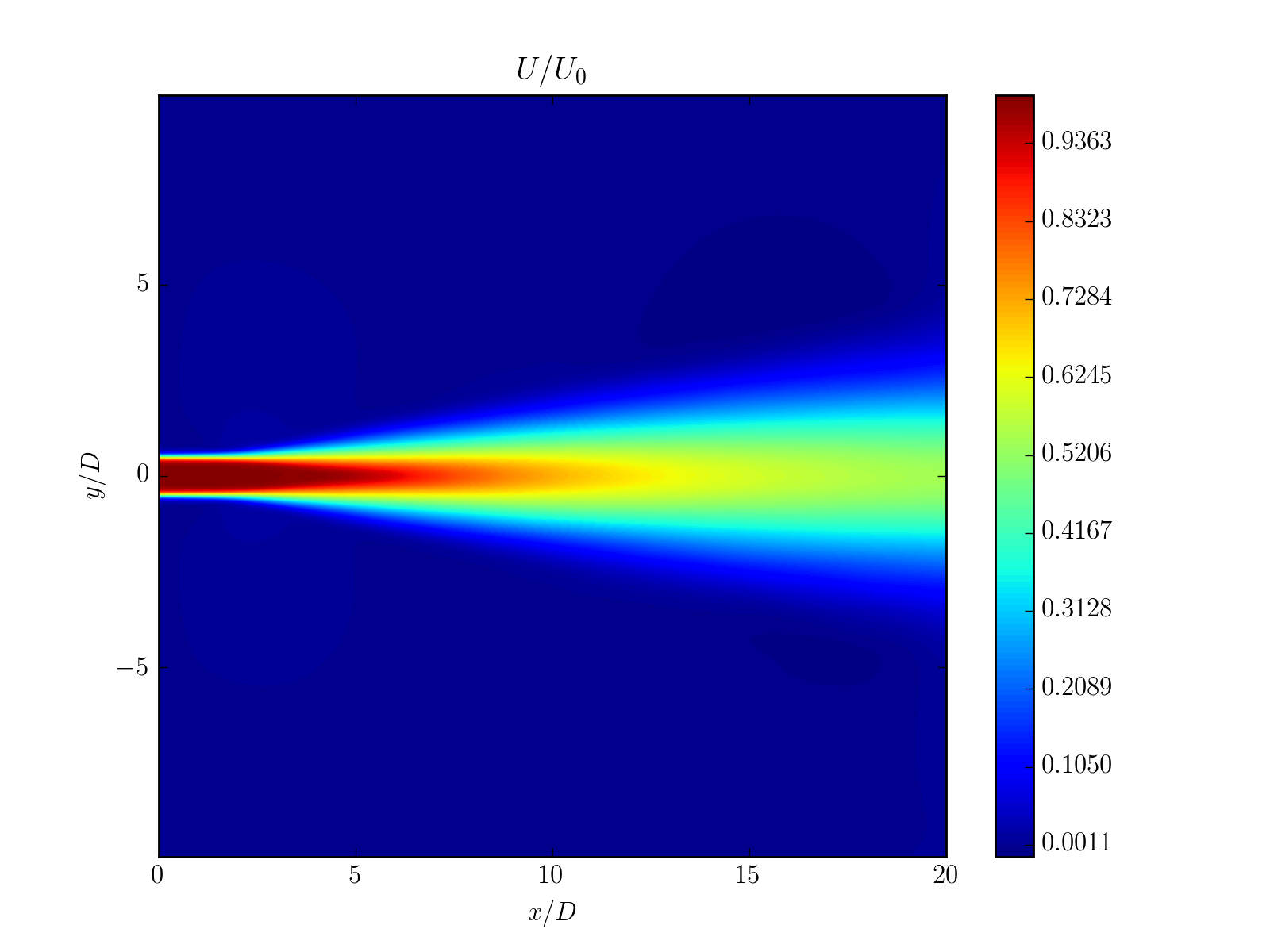}
        \caption{$U=\mean{u_1}$}
        \label{fig:DNS_U}
    \end{subfigure}
    ~ 
    \begin{subfigure}[b]{0.32\textwidth}
        \includegraphics[width=1.2\textwidth,height=1\textwidth]{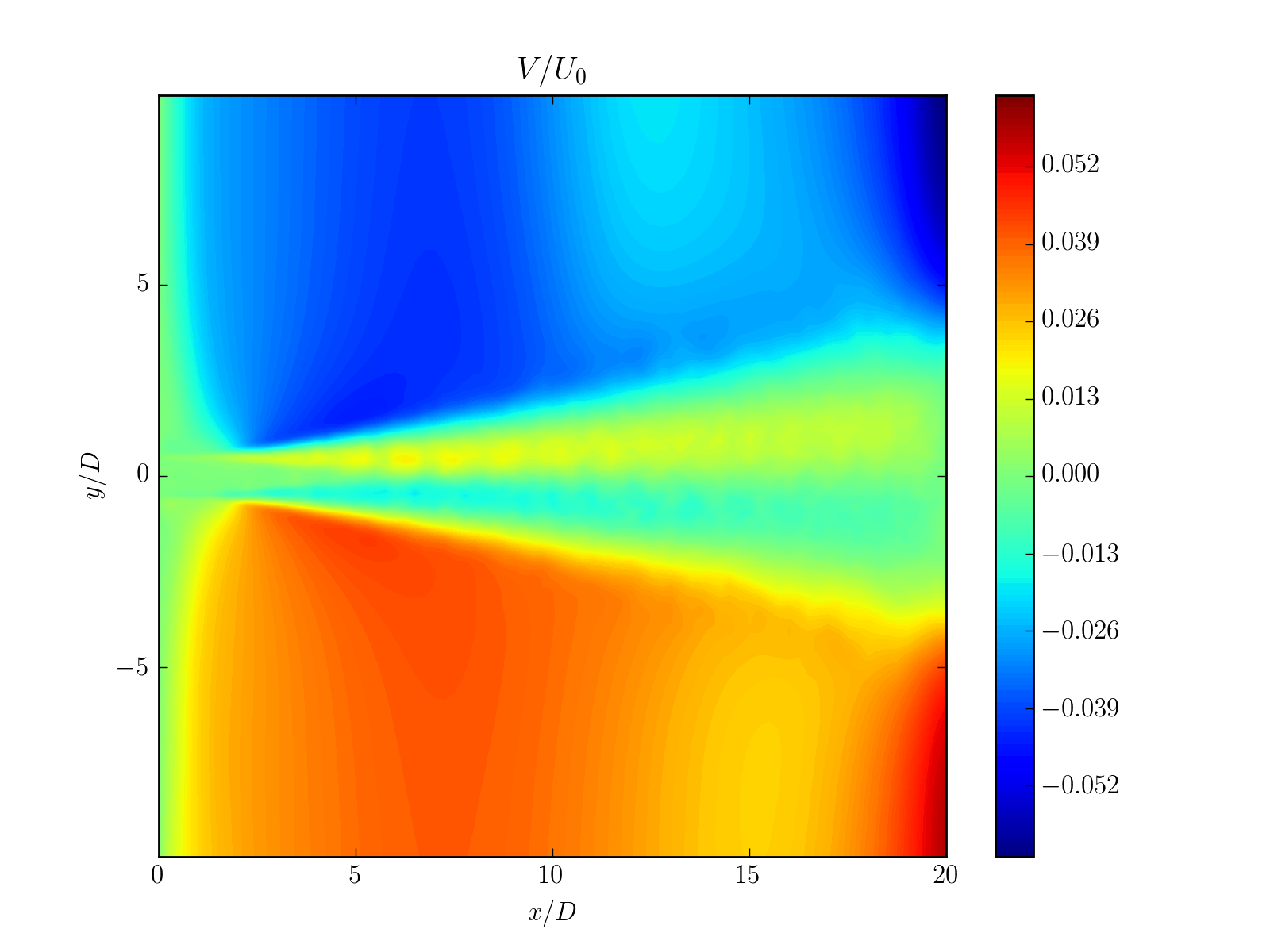}
        \caption{$V=\mean{u_2}$}
        \label{fig:DNS_V}
    \end{subfigure}
    ~ 
    \begin{subfigure}[b]{0.32\textwidth}
        \includegraphics[width=1.2\textwidth,height=1\textwidth]{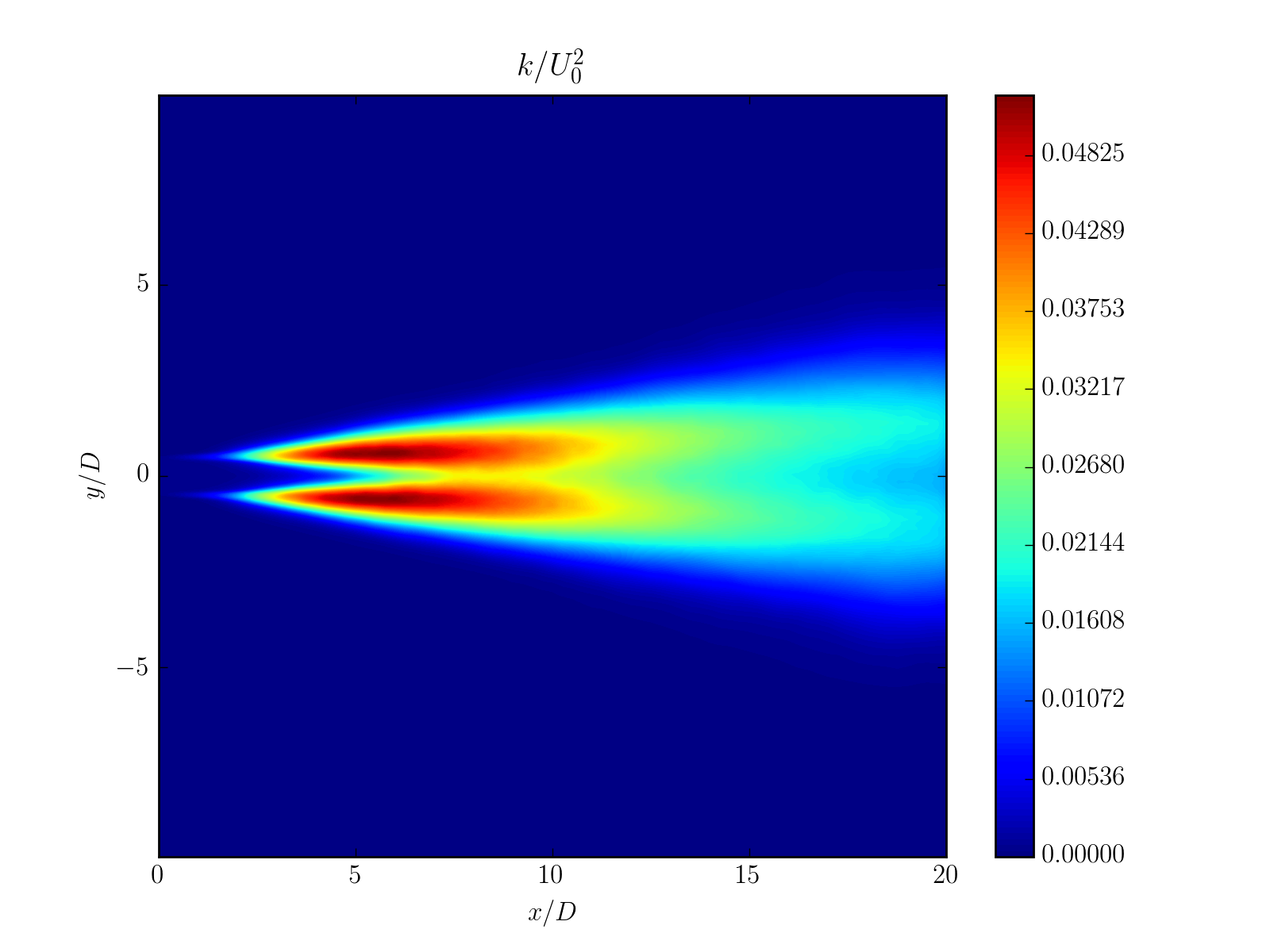}
        \caption{$k=\mean{u_i' u_i'}/2$}
        \label{fig:DNS_k}
    \end{subfigure}
  \caption{Mean velocity profile \eqref{fig:DNS_U}\eqref{fig:DNS_V} and turbulent kinetic energy \eqref{fig:DNS_k} from DNS data by \cite{klein03}}
  \label{fig:DNS_VIS}
\end{figure}

\subsection{Boundary condition}
\begin{itemize}
\item $x=0$
  \begin{equation*}
    \begin{aligned}
      U_0 &= \half + \half \tanh \left(\frac{-|y| + 0.5D}{2\theta}\right) \\
      V&=0 \\
          k &= \half \mean{u_i' u_i'} \\
          \mean{u'u'} &= \mean{v'v'} = 0.0004 \quad \textrm{at} \quad x=0, \;\; |y| < D/2 \\
          \mean{u'u'} &= \mean{v'v'} = 0 \quad \textrm{otherwise} \\
          \epsilon &= \frac{C_\mu}{\nu_t} k^2
    \end{aligned}
  \end{equation*}
  where $\theta=D/20$ is the momentum thickness. 

\item $y = \pm L_y / 2$
  \begin{equation*}
    \begin{aligned}
      (\nu_t + \nu)\left(\pdiff{U}{y}+\pdiff{V}{x}\right) &= 0 \\
      V &= 0 \\
      k &= 0 \\
      \epsilon &= 0 \\ 
    \end{aligned}
  \end{equation*}

\item $x = L_x$
  \begin{equation*}
    \begin{aligned}
      2(\nu_t + \nu)\pdiff{U}{x} -P &= 0 \\
      V &= 0 \\
      \left(\frac{\nu_t}{\sigma_k} + \nu\right)\pdiff{k}{x} &= 0 \\
      \left(\frac{\nu_t}{\sigma_\epsilon} + \nu\right)\pdiff{\epsilon}{x} &= 0 
    \end{aligned}
  \end{equation*}

\end{itemize}


\subsection{Initial condition}
Since only solutions at steady state are important, the final solution should not be sensitive to the
initial condition. However, the ``good'' initial conditions can reduce the computing time and avoid
any possible ill-poseness of simulation. One of the best initial condition is from DNS data. Alternatively, one can
make a good approximation on initial condition.
Following suggestion for the initial conditions are
not so sophisticated and based on simplification of discussion by \cite{pope2000}.

Let's assume the mean velocity is self-similar to centerline velocity and characteristic integral length scale $L$.
Now, the spread function is defined as
\begin{equation*}
S = \frac{\diff L }{\diff x} 
\end{equation*}
From the mass flux constant condition
following initial mean velocity profile can be obtained.

\begin{equation*}
  U(x,y,t=0) = \frac{D}{ D+Sx} \left[ \half + \half \tanh \left( \frac{-|y| + 0.5(D+Sx)}{(D+Sx)/10} \right) \right]
\end{equation*}
Similarly, the initial condition of fluctuation terms have been set as
\begin{equation*}
  \begin{aligned}
    \mean{u'u'} = \mean{v'v'} &=  \frac{0.0004D}{D+Sx}  &\quad \textrm{at} \quad |y| < (D+Sx)/2 \\
    \mean{u'u'} = \mean{v'v'} &= 0 &\quad \textrm{otherwise}
  \end{aligned}
\end{equation*}

Let's assume that $S = 0.06$.
The dissipation term, $\epsilon$ can be set with the assumption of initial turbulent kinetic viscosity.
\begin{equation*}
  \nu_t  (x,y,t=0) = 0.02
\end{equation*}
Hence,
\begin{equation*}
  \epsilon = \frac{C_\mu}{\nu_t} k^2
\end{equation*}

\section{Proof of Theorem \ref{thm:DR-infmMALA}}\label{apx:thm-DR-infmMALA}
\begin{proof}
With assumptions \ref{asump:eig_decay} and \ref{asump:fwd_diff} (same as Assumption 3.1 of \cite{beskos2017}), we have $\nabla_v\Phi(v)=\Co^{\half} \nabla_u\Phi(u) \in \Hi^{\kappa-\ell}$ for some $\ell \in[0,\kappa-\half)$, $\kappa>\half$.
With the approximations \eqref{eq:apx_postcov} \eqref{eq:apx_postcov2}, we directly get the assumption 3.2 of \cite{beskos2017}:
$\aKo(v)=\Vo_r D_r \Vo_r^* + \Vo_\perp\Vo_\perp^*$ is self-adjoint and positive-definite operator on $\Hi$ and it is such that
\begin{itemize}
\item[i)] $\Imag(\aKo(v)^\half) = \Imag(\Io) = \Hi$\ ;
\item[ii)] $\aKo(v)^\half (\aKo(v)^\half)^* - \Io = \Vo_r ( D_r -I_r ) \Vo_r^*$ is an operator on $\Hi$ with finite rank thus automatically Hilbert-Schmidt.
\end{itemize}
Assumption 3.3 of \cite{beskos2017} can be derived by the fact that $(\aKo(v)-\Io) v=\Vo_r ( D_r -I_r ) \Vo_r^* v \in \Hi_r \subset \Hi$, which, together with assumptions \ref{asump:eig_decay} and \ref{asump:fwd_diff}, yield a similar corollary as 3.4 of \cite{beskos2017}: $\hat g(v) \in \Hi$.
Define the following bivariate law and its reference measure
\begin{equation*}
\begin{aligned}
\nu(dv, dv') &= \mu(dv) Q(v, dv'), \quad  \mu(dv) = \mu(\Co^{-\half} du), \quad Q(v, dv') \,\textrm{being} \, \eqref{eq:infmMALA-whiten} \, \textrm{with approximations} \, \eqref{eq:apx_postcov2} \, \textrm{applied}\\
\nu_0(dv, dv') &= \mu_0(dv) Q_0(v, dv'), \quad \mu_0(dv) = \mathcal N(0, \Io), \quad Q_0(v, dv') \,\textrm{being} \, \eqref{eq:infmMALA-whiten} \, \textrm{with} \; \tilde v \; \textrm{replaced by} \; \xi \sim \mathcal N(0, \Io)
\end{aligned}
\end{equation*}
Then the Feldman-Hajek theorem \citep[Theorem 2.23 in][]{da14} can be applied to $\mathcal N\left(\frac{\sqrt h}2 \hat g(v)\right) \simeq \mathcal N(0, \Io)$, similarly as Theorem 3.5 of \cite{beskos2017}. 
The acceptance probability (c.f. Algorithm 3.7 $\infty$-mMALA in \cite{beskos2017}) can be obtained with $\Ko(u)$ replaced by $\aKo(v)$ and $g(u)$ replaced by $\hat g(v)$.
\end{proof}

\section{Proof. of Corollary \ref{cor:acpt_compare}}\label{apx:cor-acpt_compare}
\begin{proof}
We assume $\gamma_\perp=0$.
Substituting the approximations \eqref{eq:lr-apx} \eqref{eq:apx_postcov} \eqref{eq:apx_postcov2} into \eqref{eq:acpt_DRinfmMALA} yields
\begin{equation*}
\begin{aligned}
\log \lambda(w^*;v) =&
- \tfrac{h}{8} \langle -(D_r-I_r) \Vo_r^* v - D_r \gamma_r \Vo_r^* \nabla_v \Phi(v), \Lambda_r \Vo_r^* v - \gamma_r \Vo_r^* \nabla_v \Phi(v) \rangle \\
& + \tfrac{\sqrt{h}}{2} \langle \Lambda_r \Vo_r^* v - \gamma_r \Vo_r^* \nabla_v \Phi(v), \Vo_r^* w^* \rangle - \tfrac12 \langle \Vo_r^* w^*, \Lambda_r \Vo_r^* w^* \rangle 
- \tfrac12 \log |D_r| \\
=&
- \tfrac{h}{8} \Vert D_r^{\half} g_r(v) \Vert^2 
+ \tfrac{\sqrt{h}}{2} \langle g_r(v), \Vo_r^* w^* \rangle - \tfrac12 \Vert \Lambda_r^{\half} \Vo_r^* w^* \Vert^2
- \tfrac12 \log |D_r| \\
=&
- \half \left\Vert \frac{\sqrt h}{2} D_r^{\half} g_r(v) -  D_r^{-\half} \Vo_r^* w^* \right\Vert^2 + \half \Vert \Vo_r^* w^* \Vert^2 - \half \log |D_r|
\end{aligned}
\end{equation*}
where $g_r(v):= \Lambda_r \Vo_r^* v - \gamma_r \Vo_r^* \nabla_v \Phi(v)$.

On the other hand, with \eqref{eq:connection2dili}, we have
\begin{equation*}
\begin{split}
\frac{\sqrt h}{2} D_r^{\half} g_r(v) -  D_r^{-\half} \Vo_r^* w^* =&
\frac{\sqrt h}{2} D_r^{\half} \Lambda_r \Vo_r^* v -  D_r^{-\half} \Vo_r^* \frac{v'-\rho_0 v}{\rho_2} - \frac{\sqrt h}{2} D_r^{\half} \gamma_r \Vo_r^* \nabla_v \Phi(v) \\
=& (\rho_2\sqrt{D_r})^{-1} ( ( \rho_1 D_r \Lambda_r + \rho_0 ) \Vo_r^* v - \Vo_r^* v' - \rho_1 D_r \gamma_r \Vo_r^* \nabla_v \Phi(v) ) \\
=& (\rho_2\sqrt{D_r})^{-1} ( ( I_r - \rho_1 D_r ) \Vo_r^* v - \Vo_r^* v' - \rho_1 D_r \gamma_r \Vo_r^* \nabla_v \Phi(v) ) \\
=& - D_{B_r}^{-1} ( \Vo_r^* v' - D_{A_r} \Vo_r^* v - D_{G_r} \gamma_r \Vo_r^* \nabla_v \Phi(v) ) 
\end{split}
\end{equation*}
and
\begin{equation*}
\Vert \Vo_r^* w^*(v,v') \Vert^2 -\Vert \Vo_r^* w^*(v',v) \Vert^2 = \frac{\Vert \Vo_r^* v'-\rho_0 \Vo_r^* v \Vert^2 - \Vert \Vo_r^* v-\rho_0 \Vo_r^* v' \Vert^2}{\rho_2^2} = \Vert \Vo_r^* v' \Vert^2 - \Vert \Vo_r^* v \Vert^2
\end{equation*}
It completes the proof by substituting the above results into \eqref{eq:acpt_DRinfmMALA} and comparing with Equation 40 of \cite{cui16}.
\end{proof}

\section{Proof of Theorem \ref{thm:bounds}}\label{apx:thm-bounds}
\begin{proof}
Since $\Ho(v)$ is trace-class, without loss of generality, we can assume that the eigenvalues $\lambda_i \downarrow 0$ in \eqref{eq:eigv} monotonically decrease to $0$.
Denote $D=(\Lambda + \Io)^{-1}$, then
\begin{equation*}
\Ko(v) = (\Io + \Vo \Lambda \Vo^*)^{-1} = \Vo (\Io +  \Lambda)^{-1} \Vo^* = \Vo D \Vo^* = \Vo_r D_r \Vo_r^* + \Vo_\perp D_\perp \Vo_\perp^*
\end{equation*}
From \eqref{eq:reform-infmMALA} we have the difference between DR-$\infty$-mMALA and $\infty$-mMALA:
\begin{equation*}
\begin{aligned}
\Vert v'_\text{\tiny DR-$\infty$-mMALA} - v'_\text{\tiny $\infty$-mMALA} \Vert &\leq 
\rho_1 \Vert (\aKo(v)-\Ko(v))v \Vert + \rho_1 \Vert (\aKo(v)s(\gamma)-\Ko(v))\nabla_v\Phi(v) \Vert + \rho_2 \Vert (\aKo(v)^\half-\Ko(v)^\half)\xi \Vert \\
&\leq \rho_1 ( \Vert \Vo_\perp (I_\perp -D_\perp) \Vo_\perp^* \Vert \Vert v\Vert + \Vert \Vo_\perp (I_\perp\gamma_\perp -D_\perp) \Vo_\perp^* \Vert \Vert \nabla_v\Phi(v)\Vert ) + \rho_2 \Vert \Vo_\perp (I_\perp -D_\perp^\half) \Vo_\perp^* \Vert \Vert \xi\Vert \\
&\leq \rho_1 (\max (I_\perp -D_\perp) \Vert v\Vert + \max |I_\perp\gamma_\perp -D_\perp| \Vert \nabla_v\Phi(v)\Vert) + \rho_2 \max (I_\perp-D_\perp^\half) \Vert \xi\Vert \\
&\leq
\begin{cases}
\rho_1 \frac{\lambda_{r+1}}{\lambda_{r+1}+1}( \Vert v\Vert + \Vert \nabla_v\Phi(v)\Vert ) + \rho_2 \frac{\lambda_{r+1}}{\lambda_{r+1}+1 + \sqrt{\lambda_{r+1}+1}} \Vert \xi\Vert , & \gamma_\perp=1\\
\rho_1 \left( \frac{\lambda_{r+1}}{\lambda_{r+1}+1} \Vert v\Vert + \Vert \nabla_v\Phi(v)\Vert \right) + \rho_2 \frac{\lambda_{r+1}}{\lambda_{r+1}+1 + \sqrt{\lambda_{r+1}+1}} \Vert \xi\Vert, & \gamma_\perp=0
\end{cases}
\end{aligned}
\end{equation*}

As seen from Section \ref{sec:DR-infmMALA}, the proposal of DILI can be derived from \eqref{eq:reform-infmMALA} by replacing $\Ko(v)$ with $\tilde\Ko(v):=\Vo_r K_r \Vo_r^* + \Vo_\perp \Vo_\perp^*$.
Note immediately $\tilde\Ko(v)^\half=\Vo_r K_r^\half \Vo_r^* + \Vo_\perp \Vo_\perp^*$. Now we have
\begin{equation*}
\begin{aligned}
\Vert v'_\text{\tiny DR-$\infty$-mMALA} - v'_\text{\tiny DILI} \Vert &\leq 
\rho_1 \Vert (\aKo(v)-\tilde\Ko(v))v \Vert + \rho_1 \Vert (\aKo(v)-\tilde\Ko(v))s(\gamma)\nabla_v\Phi(v) \Vert + \rho_2 \Vert (\aKo(v)^\half-\tilde\Ko(v)^\half)\xi \Vert \\
&\leq \rho_1 \Vert \Vo_r (D_r-K_r) \Vo_r^* \Vert (\Vert v\Vert + \Vert \nabla_v\Phi(v)\Vert ) + \rho_2 \Vert \Vo_r (D_r^\half-K_r^\half) \Vo_r^* \Vert \Vert \xi\Vert \\
&\leq \rho_1 \Vert D_r-K_r \Vert_2 (\Vert v\Vert + \Vert \nabla_v\Phi(v)\Vert ) + \rho_2 \Vert D_r^\half-K_r^\half \Vert_2 \Vert \xi\Vert
\end{aligned}
\end{equation*}

Lastly, let's look at the difference between proposals of DR-$\infty$-mHMC and $\infty$-mHMC.
Denote $v_i^\text{\tiny DR}$ and $v_i$ as the $i$-th step parameter updates in DR-$\infty$-mHMC and $\infty$-mHMC respectively,
and $\tilde v_i^\text{\tiny DR}$ and $\tilde v_i$ as their corresponding $i$-th step auxiliary updates.
For simplicity we assume $\gamma_\perp=1$.
\begin{equation*}
\begin{aligned}
\Vert v_{i+1}^\text{\tiny DR} - v_{i+1} \Vert \leq& 
\cos\eps \Vert v_i^\text{\tiny DR} - v_i \Vert  + \sin\eps \left(\Vert \tilde v_i^\text{\tiny DR} - \tilde v_i \Vert + \frac{\eps}2 \Vert \hat g(v_i^\text{\tiny DR})-g(v_i)\Vert \right)\\
\leq& \cos\eps \Vert v_i^\text{\tiny DR} - v_i \Vert  + \sin\eps \left(\Vert \tilde v_i^\text{\tiny DR} - \tilde v_i \Vert + \frac{\eps}2 (\Vert \hat g(v_i^\text{\tiny DR})-g(v_i^\text{\tiny DR})\Vert+\Vert g(v_i^\text{\tiny DR})-g(v_i)\Vert) \right)\\
\leq& \left(\cos\eps + \frac{L\eps\sin\eps}2 \right) \Vert v_i^\text{\tiny DR} - v_i \Vert + \frac{\lambda_{r+1}\eps\sin\eps}{2(\lambda_{r+1}+1)}( \Vert v_i^\text{\tiny DR}\Vert + \Vert \nabla_v\Phi(v_i^\text{\tiny DR})\Vert ) + \sin\eps \Vert \tilde v_i^\text{\tiny DR} - \tilde v_i \Vert \\
\Vert \tilde v_{i+1}^\text{\tiny DR} - \tilde v_{i+1} \Vert \leq& 
\sin\eps \Vert v_i^\text{\tiny DR} - v_i \Vert  + \cos\eps \left(\Vert \tilde v_i^\text{\tiny DR} - \tilde v_i \Vert + \frac{\eps}2 \Vert \hat g(v_i^\text{\tiny DR})-g(v_i)\Vert \right) + \frac{\eps}2 \Vert \hat g(v_{i+1}^\text{\tiny DR})-g(v_{i+1})\Vert\\
\leq& \left(\sin\eps + \frac{L\eps\cos\eps}2 \right) \Vert v_i^\text{\tiny DR} - v_i \Vert + \frac{\lambda_{r+1}\eps\cos\eps}{2(\lambda_{r+1}+1)}( \Vert v_i^\text{\tiny DR}\Vert + \Vert \nabla_v\Phi(v_i^\text{\tiny DR})\Vert ) + \cos\eps \Vert \tilde v_i^\text{\tiny DR} - \tilde v_i \Vert \\
& + \frac{L\eps}2  \Vert v_{i+1}^\text{\tiny DR} - v_{i+1} \Vert + \frac{\lambda_{r+1}\eps}{2(\lambda_{r+1}+1)}( \Vert v_{i+1}^\text{\tiny DR}\Vert + \Vert \nabla_v\Phi(v_{i+1}^\text{\tiny DR})\Vert )
\end{aligned}
\end{equation*}

Denote $d_{i+1}:=\Vert v_{i+1}^\text{\tiny DR} - v_{i+1} \Vert$ and $\tilde d_{i+1}:=\Vert \tilde v_{i+1}^\text{\tiny DR} - \tilde v_{i+1} \Vert$. Then we have
\begin{equation*}
d_{i+1} \leq C_0 d_i + \mathcal O(\lambda_{r+1}) + C_1 \tilde d_i, \quad \tilde d_i \leq C_2 d_{i-1} + \mathcal O(\lambda_{r+1}) + C_3 \tilde d_{i-1} + C_4 d_i
\end{equation*}
Noticing that $d_0=\tilde d_0=0$, thus we can conclude that
\begin{equation*}
d_{I} = \Vert v_I^\text{\tiny DR} - v_I \Vert \leq \mathcal O(\lambda_{r+1}) \ .
\end{equation*}

\end{proof}

\end{document}